\newif\ifprocs
\procstrue
\procsfalse 

\ifprocs
\documentclass[a4paper,USenglish,cleveref,autoref,thm-restate]{lipics-v2021}
\else
\documentclass{article}
\usepackage[utf8]{inputenc}
\usepackage{fullpage}
\fi

\usepackage{setspace}
\usepackage{graphicx}

\usepackage{url} 

\usepackage{soul}

\usepackage{enumitem}
\usepackage{microtype}
\usepackage{lineno}

\usepackage{amsthm,amsmath,amssymb}
\usepackage{thmtools}
\usepackage{thm-restate}
\usepackage{mathtools}
\usepackage{csquotes}
\MakeOuterQuote{"}

\usepackage{xcolor}
\ifprocs
\hypersetup{colorlinks=true,allcolors=blue}
\else
\usepackage[colorlinks=true,allcolors=blue]{hyperref}
\usepackage{xspace}
\usepackage{xfrac}
\fi

\usepackage{color}
\usepackage[ruled,vlined]{algorithm2e}
\usepackage{tabularx}
\usepackage{pdflscape,latexsym,bm,array,caption,textcomp}

\newcolumntype{e}{>{\centering\arraybackslash}p{10.6em}}
\newcolumntype{w}{>{\centering\arraybackslash}p{15.5em}}
\newcolumntype{t}{>{\centering\arraybackslash}p{14.5em}}

\ifprocs
\newtheorem*{theorem*}{Theorem}

\newtheorem{conjecture}{Conjecture}

\else
\newtheorem{theorem}{Theorem}[section]
\newtheorem{corollary}[theorem]{Corollary}
\newtheorem{lemma}[theorem]{Lemma}
\newtheorem{claim}[theorem]{Claim}
\newtheorem{conjecture}{Conjecture}
\newtheorem{proposition}[theorem]{Proposition}
\theoremstyle{definition}
\newtheorem{definition}{Definition}[section]
\newtheorem{observation}[theorem]{Observation}
\fi

\theoremstyle{remark}
\newtheorem{remark}{Remark}[section]

\newcommand{\Samp}{\textsc{SAMP}}
\newcommand{\PRSamp}{\textsc{PR-SAMP}}

\newcommand{\Dual}{\textsc{DUAL}}
\newcommand{\Eval}{\textsc{EVAL}}
\newcommand{\SEval}{\textsc{SET-EVAL}}
\newcommand{\CO}{\textsc{COND}}
\newcommand{\CD}{\textsc{COND-EVAL}}
\newcommand{\CE}{\textsc{COND-PR}}
\newcommand{\CEval}{\textsc{CEVAL}}
\newcommand{\supp}{\textsc{supp}}
\newcommand{\GHD}{\textsc{Gap-Hamming}}
\newcommand{\rk}{\textsc{rank}}
\newcommand{\blk}{\textsc{blk}}
\newcommand{\E}{\mathbb{E}}

\newcommand{\PCo}{\textsc{PAIRCOND}}
\newcommand{\ksmedits}[1]{{#1}}

\title{Support Size Estimation: The Power of Conditioning}

\ifprocs

\authorrunning{J. Open Access and J.\,R. Public} 

\else
\author{Diptarka Chakraborty\footnote{National University of Singapore, Singapore. Supported in part by NUS ODPRT Grant, WBS No. R-252-000-A94-133. Email: diptarka@comp.nus.edu.sg} \and Gunjan Kumar \footnote{National University of Singapore, Singapore.  Email: dcsgunj@nus.edu.sg} \and Kuldeep S. Meel\footnote{National University of Singapore, Singapore. Supported in part by  National Research Foundation Singapore under its NRF Fellowship Programme[NRF-NRFFAI1-2019-0004 ] , Ministry of Education Singapore Tier 2 grant MOE-T2EP20121-0011, and Ministry of Education Singapore Tier 1 Grant [R-252-000-B59-114 ].   Email: meel@comp.nus.edu.sg}}
  
\fi

\ifprocs
\Copyright{Jane Open Access and Joan R. Public} 



\EventEditors{John Q. Open and Joan R. Access}
\EventNoEds{2}
\EventLongTitle{42nd Conference on Very Important Topics (CVIT 2016)}
\EventShortTitle{ICALP 2022}
\EventAcronym{CVIT}
\EventYear{2016}
\EventDate{December 24--27, 2016}
\EventLocation{Little Whinging, United Kingdom}
\EventLogo{}
\SeriesVolume{42}
\ArticleNo{23}
\else
\fi

\usepackage{multirow}
\ifprocs
\else
\fi

\usepackage[colorinlistoftodos,prependcaption,textsize=tiny]{todonotes}
\usepackage{makecell}
\date{}
\begin{document}
\pagenumbering{gobble}
\maketitle

\begin{abstract}
We consider the problem of estimating the support size of a distribution $D$. Our investigations are pursued through the lens of distribution testing and seek to understand the power of conditional sampling (denoted as {\CO}), wherein one is allowed to query the given distribution conditioned on an arbitrary subset $S$. The primary contribution of this work is to introduce a new approach to lower bounds for the {\CO} model that relies on using  powerful tools from information theory and communication complexity. 

Our approach allows us to obtain surprisingly strong lower bounds for the {\CO} model and its extensions.
 
 \begin{itemize}
 
 \item We bridges the longstanding gap between  the upper ($O(\log \log n + \frac{1}{\epsilon^2})$) and the lower bound $\Omega(\sqrt{\log \log n})$ for {\CO} model by providing a nearly matching lower bound. Surprisingly, we show that even if we get to know the actual probabilities along with {\CO} samples, still $\Omega(\log \log n + \frac{1}{\epsilon^2 \log (1/\epsilon)})$ queries are necessary.
    
\item We obtain the first non-trivial lower bound for {\CO} equipped with an additional oracle that reveals the conditional probabilities of the samples (to the best of our knowledge, this subsumes all of the models previously studied): in particular, we demonstrate that  $\Omega(\log \log \log n + \frac{1}{\epsilon^2 \log (1/\epsilon)})$ queries are necessary.
    
%
 \end{itemize}
\end{abstract}
\newpage
\pagenumbering{arabic}

\section{Introduction}\label{sec:introduction}

We consider the problem of estimating the support size of a distribution $D$ over a domain $\Omega$ (of size $n$), which is defined as follows:
\begin{align*}
    \supp(D) := \{x \mid D(x) > 0\}.
\end{align*}
We are interested in $(\epsilon,\delta)$-approximation\footnote{We want to estimate $|\supp(D)|$ by $\hat{s}$ such that $\frac{|\supp(D)|}{(1+\epsilon)} \le \hat{s} \le (1+\epsilon)  |\supp(D)|$ with the success probability  at least $1-\delta$. This version is also referred to as $(1+\epsilon)$-multiplicative factor estimation. Another interesting version to consider is the additive $\epsilon n$-estimation (where $n$ denotes the size of the domain) which asks to output a $\hat{s}$ such that $|\supp(D)| - \epsilon n \le \hat{s} \le  |\supp(D)| + \epsilon n$ with the success probability  at least $1-\delta$. Unless otherwise stated explicitly, we consider the multiplicative variant.} of the size of $\supp (D)$ (i.e., $|\supp(D)|$). For simplicity in exposition, throughout this paper, we consider $\delta$ to be a small constant (more specifically, 1/3, which can be reduced to any arbitrary small constant.). The support size estimation is a fundamental problem in data science and finds a myriad of applications ranging from database management, biology, ecology, genetics, linguistics, neuroscience, and physics (see~\cite{valiant2011estimating} and the references therein). Naturally, the distribution is not specified explicitly, and therefore, the complexity of the problem depends on the queries that one is allowed to the distribution. As such, the primary objective is to minimize the number of queries (aka \emph{query complexity}).  

Along with support size estimation, several other properties of distributions have attracted investigations over the past three decades (see~\cite{canonne2020survey}). As such, several query models have been considered by the research community. The simplest model {\Samp} only allows drawing independent and identically distributed samples from $D$. Valiant and Valiant~\cite{valiant2011estimating} showed that to get an estimation up to an additive factor of $\epsilon n$ (for any $\epsilon > 0$), $O(n/\epsilon^2 \log n)$ samples suffice, which was subsequently improved to $O(\frac{n}{\log n}\log^2 (1/\epsilon))$ by Wu and Yang~\cite{wu2019chebyshev}. Further, Wu and Yang proved that $\Omega(\frac{n}{\log n}\log^2 (1/\epsilon))$ samples are also necessary to get an estimate up to an additive error of $\epsilon n$. A natural extension to {\Samp} is called \emph{probability-revealing sample} or {\PRSamp}, due to Onak and Sun~\cite{onak2018probability}, wherein instead of just returning an independent sample $x$ from $D$ (as in {\Samp}), the oracle provides a pair $(x, D(x))$ (i.e., a sample along with the probability assigned on it by $D$). Onak and Sun showed that to estimate the support size up to an additive error of $\epsilon n$, $\Theta(1/\epsilon^2)$ samples are necessary and sufficient in the {\PRSamp} model. The same upper bound for the {\PRSamp} model was also implicit in the work by Canonne and Rubinfeld~\cite{canonne2014aggregate}.

As we seek to explore more powerful models than {\PRSamp}, a model of interest is {\Dual}~\cite{batu2005complexity, guha2009sublinear, canonne2014aggregate} wherein we have access to two oracles: One is {\Samp} that provides a sample from $D$, and another is {\Eval} that given any $x\in \Omega$, outputs the value of $D(x)$. In the {\Dual} model, for any $\epsilon_1,\epsilon_2 \in (0,1]$, distinguishing between whether the support size of $D$ is at most $\epsilon_1 n$ or at least $\epsilon_2 n$ requires $\Theta(1/(\epsilon_2 -\epsilon_1)^2)$ queries~\cite{canonne2014aggregate}. An extension of {\Eval} is {\CEval} wherein for totally ordered domains, given $x$, {\CEval} outputs $\sum_{y \preccurlyeq x} D(x)$. Similarly, \textsc{CDUAL} is an extension of {\Dual} where we have access to oracles {\Samp} and {\CEval}. Caferov et al.~\cite{caferov2015optimal} showed that $\Omega(\frac{1}{\epsilon^2})$  queries are needed  in the \textsc{CDUAL} model to estimate the support size up to an additive factor of $\epsilon n$. However, to the best of our knowledge, no non-trivial result is known for the support size estimation problem with $(1+\epsilon)$ multiplicative error in the above models.

While {\Samp}, {\PRSamp}, and {\Dual} are natural models, they are limiting in theory and practice as they fail to capture several scenarios wherein one is allowed more powerful access to the distribution under consideration. Accordingly, Chakraborty \emph{et al.}~\cite{chakraborty2016power} and Canonne \emph{et al.}~\cite{canonne2014testing} initiated the study of a more general sampling model {\CO}, where we are allowed to draw samples conditioning on any arbitrary subsets of the domain $\Omega$. More specifically, the sampling oracle takes a subset $S \subseteq \Omega$ chosen by the algorithm and returns an element $x\in S$ with probability $D(x)/D(S)$ if $D(S)>0$. The models proposed by Chakraborty et al. and Canonne et al. differ in their behavior for the case when $D(S) = 0$.  
The model proposed by Chakraborty \emph{et al.}~\cite{chakraborty2016power} allows the  oracle to return a uniformly random element from $S$ when $D(S)=0$. On the other hand, the {\CO} model defined in Canonne \emph{et al.}~\cite{canonne2014testing} assumes that the oracle (and hence the algorithm) returns "failure" and terminates if $D(S)=0$\footnote{
Note that analyzing each step of the algorithm makes it possible to determine the set,  conditioning on which caused the algorithm to output "failure" and terminate. The rest of the algorithm can then execute with the information that $D(S) = 0$ for the above set $S$. Therefore, for the simplicity of exposition, we will assume that the {\CO} model defined in Canonne \emph{et al.} returns "failure" when $D(S)=0$ but the algorithm does not terminate.}. Note that the {\CO} model of Canonne \emph{et al.} is more powerful than that of Chakraborty \emph{et al.} since, when $D(S) = 0$, in the former case,  we get to know that $D(S) = 0$ whereas in the latter case, we get a uniformly random element of $S$. 

The relative power of the {\CO} model of Canonne {\em et al.} over that of Chakraborty {\em et al.} is also exhibited in the context of support size estimation. Acharya, Canonne, and Kamath~\cite{acharya2015chasm} designed an algorithm with query complexity $\Tilde{O}(\log \log n / \epsilon^3)$ in the {\CO} model of Chakraborty \emph{et al.} to estimate the support size up to $(1+\epsilon)$ multiplicative factor under the assumption that the probability of each element is at least $\Omega(1/n)$. They also note that the assumption of a lower bound on the probability of each element is required for their techniques to work. Surprisingly, a result of Falahatgar et al.~\cite{falahatgar2016estimating} implies that $O(\log \log n + \frac{1}{\epsilon^2})$ queries are sufficient  for {\CO} model of Canonne \emph{et al.} for any arbitrary probability distribution, i.e., there is no requirement for the assumption of lower bound on the probability of each element\footnote{They considered oracle access in which, given any $S \subseteq \Omega$, it can be determined if $S \cap \supp(D) = \emptyset$ or not. Clearly, the {\CO} model of Canonne \emph{et al.} generalizes this model. It is worth noting that in the {\Samp} model, to get any meaningful upper bound, we need the assumption that each element has a probability at least $1/n$; however, such an assumption is not needed in the {\CO} model or its extensions.}. On the lower bound side, we only know that at least $\Omega(\sqrt{\log \log n})$ {\CO} queries are necessary~\cite{chakraborty2016power}.

The two models were introduced in the context of uniformity testing, wherein the choice of how to handle the case of $D(S)=0$ did not make any significant differences. We would like to emphasize that Canonne et al.'s model is more powerful than that of Chakraborty et al., and thus any lower bound shown in the first one also provides the same in the latter one. Moreover, the model of Canonne et al. closely approximates the behavior of modern probabilistic programming systems~\cite{GHNR14}. Therefore, throughout this paper, we consider the {\CO} model of Canonne et al.

Since its introduction, the {\CO} model has attained significant attention both in theory and practice. From a theoretical perspective, various other distribution testing problems have been studied under the {\CO} model~\cite{falahatgar2015faster, kamath2019anaconda, Narayanan21} and its variant like \emph{subcube conditioning model}~\cite{bhattacharyya2018property, canonne2021random, chen2021learning}. Apart from that, the {\CO} model and its variants find real-world applications in the areas like formal methods and machine learning (e.g.,~\cite{chakraborty2019testing, meel2020testing, GJM22}). Also, the modern probabilistic programming systems extend classical programs with the addition of sampling and \emph{observe}, where the semantics of the observe match that of Canonne et al.'s {\CO} model~\cite{GHNR14}.

It is worth remarking that the {\CO} model is incomparable with {\PRSamp}. Therefore, it is quite natural, both from a theory and practical perspective, to consider a sampling model that inherits power from both {\CO} and {\PRSamp}. We consider a model where we are allowed to condition on any arbitrary subset $S\subseteq \Omega$, and if $D(S)>0$, we receive a sample $x \in S$ with probability $D(x)/D(S)$ (as in {\CO}) along with the probability assigned on it by $D$ (i.e., $D(x)$); "failure" otherwise. We refer to this model as \emph{probability-revealing conditional sample}\footnote{The name is motivated from the {\PRSamp} model~\cite{onak2018probability}.} or in short {\CE}. To the best of our knowledge, Golia, Juba, and Meel~\cite{GJM22} were the first ones to initiate the study of the {\CE} model. Their work focused on the multiplicative estimation of entropy on the {\CE} model. Golia et al. were primarily motivated to investigate the {\CE} model upon the observation that the usage of the model counter and a sampler can simulate the {\CE} model wherein a circuit specifies the distribution. Also, implicit in their study is that the availability of model counter and samplers~\cite{DM22} allows one to simulate generalization of {\CE} model wherein for a given input $D$ and $S$ in addition to $D(x)$ for a sampled item $x \in S$, the oracle also returns the value of $D(x)/D(S)$ (the conditional probability of $x$ given $S$). We refer to this model as  {\em conditional sampling evaluation model}\footnote{The name is motivated from the standard \emph{evaluation} model {\Eval}~\cite{rubinfeld2009testing} where given any $x\in \Omega$, we get the value of the probability density function of $D$ at $x$.}, or in short {\CD}. To the best of our knowledge, the {\CD} subsumes all the previously studied variants of the {\CO} model (see Figure~\ref{fig:models}, and Section~\ref{sec:power-SEval} where we provide a few examples showcasing the power of {\CD}). 
 
 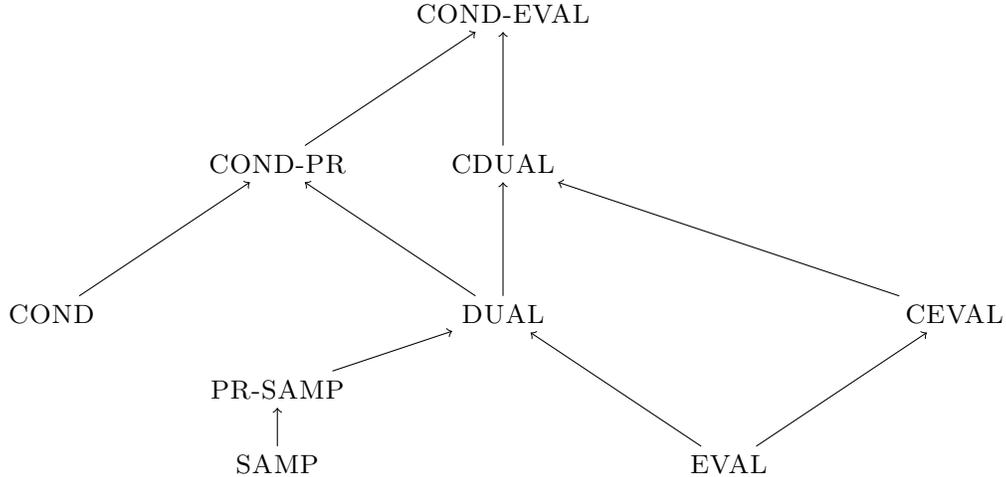
\begin{figure}[!htb]
 \centering
	\begin{tikzpicture}
	\node[] (samp) at (0,0) {{\Samp}};
	\node[] (samp-pr) at (0,1) {{\PRSamp}};

	\node[] (eval) at (6,0) {{\Eval}};
	\node[] (dual) at (3,2) {{\Dual}};
	\node[] (cond) at (-3,2) {{\CO}};
	
	\node[] (ceval) at (9,2) {{\CEval}};
	\node[] (cdual) at (3,4) {\textsc{CDUAL}};
	\node[] (cond-pr) at (0,4) {{\CE}};
	\node[] (cond-eval) at (3,6) {{\CD}};

\draw[->] (eval) -- (dual);
\draw[->] (eval)--(ceval);
\draw[->] (dual)--(cond-pr);
\draw[->] (cond)--(cond-pr);
\draw[->] (dual)--(cdual);
\draw[->] (ceval)--(cdual);
\draw[->] (cond-pr)--(cond-eval);
\draw[->] (cdual)--(cond-eval);

\draw[->] (samp)--(samp-pr);
\draw[->] (samp-pr)--(dual);

	\end{tikzpicture}
	\caption{Relative power of different models: An edge $u \to v$ means the model $v$ is more powerful than the model $u$.}
	\label{fig:models}
\end{figure}



 \ksmedits{
To summarize, there has been a long line of research that has relied on the usage of the {\CO} model and its variants, resulting in significant improvements in the query complexity for several problems in distribution testing. While there has been a multitude of techniques for obtaining upper bounds for the {\CO} model and its variants, such has not been the case for lower bounds. In particular, the  prior techniques developed in the context of support size estimation for {\CO} model have primarily relied on the observation that an algorithm $\mathcal{A}$ that makes $q$ {\CO} queries can be simulated by a decision tree of $O(q 2^{q^2})$ nodes. Accordingly, the foregoing observation allows one to obtain $\Omega(\sqrt{\log \log n})$ lower bound for {\CO}, which leaves open a major gap with respect to the upper bound of $O(\log \log n + \frac{1}{\epsilon^2})$.  The situation is even more dire when considering models that augment {\CO} with powerful oracles such as {\Eval} since the approach based on decision tree fails due to additional information supplied by oracles such as {\Eval}, and accordingly, no non-trivial lower bounds are known for models such as  \textsc{CDUAL}, {\CD}, and the like. Therefore, there is a desperate need for new lower bound techniques  to fully understand the power of {\CO} model and its natural extensions. 

 
 \subsection{Our Contribution}
 
One of our primary contributions is to provide a seemingly new approach to proving lower bounds in the {\CO} model and its more powerful variants. Our approach is based on information theory and reductions to problems in communication complexity. We note that the communication complexity-based approaches to lower bounds have been explored in prior work; such approaches are only limited to weaker models such as {\Samp} and {\PCo} models~\cite{blais2019distribution,blais2012property}. While we demonstrate the application of our approach in the context of support size estimation for different variants of {\CO}, we believe our approach is of general interest and can be applied to other distribution testing problems.


For ease of exposition, here we situate the discussion in the context of the most general model, {\CD}.
One of the inherent difficulties in proving any non-trivial lower bound for the {\CD} model arises from the fact that the different sets for conditioning can overlap in an arbitrary manner (and unlike in {\PCo} model, these sets are of arbitrary size) and further be chosen in an adaptive way. 
The adaptivity and arbitrary size of sets make it 
 extremely difficult to upper bound the conditional entropy at any step of the algorithm. Furthermore, the power of revealing the probability mass (on any set) by a {\CD} query risks licking "a lot of" information which makes it even more challenging. The key departure from earlier work is the choice of an infinite family of distributions, so the range of outcomes of an algorithm is continuous and so cannot be encoded with any finite (or even infinite) length message. To this end, we rely on Fano's inequality, a fundamental tool in information theory, to show lower bounds for statistical estimation. In order to apply Fano's inequality, we need to upper bound the information gain at every step of the algorithm. Our approach proceeds by relying on a restricted model of conditioning where the queried sets are {\em laminar}, i.e., either they do not intersect or are subsets/supersets of each other. Accordingly, we first obtain lower bounds for the restricted model and lift to the lower bounds for the {\CD} model.  


} 

 Our approach is compelling enough to provide non-trivial lower bounds in the most potent {\CD} model, for which no lower bound was known before. However, before providing the usefulness of the general framework, let us demonstrate how a special instantiation of our approach can be applied to obtain strong lower bounds in the context of support size estimation for the {\CO} model and its (simpler) variants. 
 Our first result 
 bridges the long-standing gap between the upper bound of $O(\log \log n + \frac{1}{\epsilon^2})$ and the lower bound of $\Omega(\sqrt{\log \log n})$ in case of {\CO} model.  In particular, we obtain an $\Omega(\log \log n + \frac{1}{\epsilon^2 \log (1/\epsilon)})$ lower bound on the query complexity in the {\CE} model, which in turn provides the same lower bound for the  {\PRSamp}, {\Dual},  and {\CO} model. 

\begin{restatable}{theorem}{condevalLB}
\label{thm:cond-eval-LB}
Every algorithm that, given {\CE} access to a distribution $D$ on $[n]$ and $\epsilon \in (0,1]$, estimates the support size $|\supp(D)|$ within a multiplicative $(1+\epsilon)$-factor with probability at least $\frac{2}{3}$, must make $\Omega(\log \log n + \frac{1}{\epsilon^2 \log (1/\epsilon)})$ queries to the {\CE} oracle.
\end{restatable}

Recall the best-known upper bound for the support size estimation in {\CO}
is $O(\log \log n + \frac{1}{\varepsilon^2})$ and therefore, the above theorem achieves a near-matching lower bound in the context of {\CO} and {\CE}. It is rather surprising that the combination of {\PRSamp} and {\CO} does not yield power in the context of the support size estimation. It is worth mentioning that we already know of a lower bound of $\Omega(1/\epsilon^2)$ for the {\PRSamp} and {\Dual} model due to~\cite{canonne2014aggregate}. Hence, our result along with~\cite{canonne2014aggregate} provides a lower bound of $\Omega(\log \log n + \frac{1}{\epsilon^2})$ for these two models (i.e., we can get rid of the annoying $\log (1/\epsilon)$ factor from the lower bound term of~\autoref{thm:cond-eval-LB}).

\ksmedits{Our primary result is to establish the first non-trivial lower bound in the context of {\CD},  
 which in turn provides the first-known lower bound for many other previously studied models, such as \textsc{CDUAL}.
 }

\begin{restatable}{theorem}{celowerbound}
\label{thm:CDlowerboundfull}
Any  algorithm that, given {\CD} access to a distribution $D$ on $[n]$ approximates the support size $|\supp(D)|$ within a multiplicative $(1+\epsilon)$-factor with  probability at least $2/3$,  must make $\Omega(\log \log \log n +\frac{1}{\epsilon^2 \log (1/\epsilon) })$  queries.
\end{restatable}

Since {\CD} subsumes {\CE}, we have  the upper bound of $O(\log \log n +\frac{1}{ \epsilon^2})$, and the ensuing gap leaves open an interesting question. \ksmedits{It is worth remarking that it is not hard to see (by extending the proof of~\autoref{thm:cond-eval-LB}) that the upper bound is nearly tight if one were to replace {\CD} with approximate-{\CD} wherein, for a given $D$ and $S$, the oracle essentially provides an estimate of $D(S)$ up to a small multiplicative error} (see the discussion at the end of Section~\ref{sec:cd-LB}). We conjecture that the upper bound is tight for {\CD}. 

\begin{conjecture}\label{conj:cdlowerbound}
Any  algorithm that, given {\CD} access to a distribution $D$ on $[n]$ approximates the support size $|\supp(D)|$ within a multiplicative $(1+\epsilon)$-factor with  probability at least $2/3$,  must make $\Omega(\log \log n +\frac{1}{\epsilon^2 \log (1/\epsilon) })$  queries.
\end{conjecture}

 The validity of the above conjecture would establish the significant power of the {\CO} model in the context of support size estimation as {\CD} and {\CE}, despite being augmented with powerful oracles in addition to conditioning, do not yield better algorithms.

\subsection{Technical Overview}
\label{sec:tec}

\paragraph*{Lower bound for {\CE}. }
We start with deriving the lower bound of~\autoref{thm:cond-eval-LB}. It consists of two parts -- an $\Omega(\log \log n)$ lower bound for a multiplicative $4/3$-factor estimation algorithm and an $\Omega(\frac{1}{\epsilon^2 \log \frac{1}{\epsilon}})$ lower bound for an additive $\epsilon n$-factor algorithm. 
We first show an $\Omega(\log \log n)$ lower bound for a multiplicative $4/3$-factor algorithm (\autoref{lemma:loglogn}). For that purpose, we consider an \emph{integer-guessing} game between Alice and Bob, where Alice uniformly at random chooses an integer $x \in [\log n]$. Then sends a message (binary string) to Bob. Upon receiving the message, Bob's task is to guess $x$ correctly (with high probability). Since the entropy of the chosen integer is $\log \log n$, by Shannon's source coding theorem, the length of the message, on average, must be $\Omega(\log \log n)$. We show that if there exists an algorithm $T$ that makes $t$ {\CE} queries then it suffices for Alice to send a message of length $O(t)$, and hence $t = \Omega(\log \log n)$. 

To show the same, for each $x \in [\log n]$, Alice considers a distribution $D_x$ with support $[2^x]$. The probability $D_x(j)$ of  an element $j \in [2^x]$ in $D_x$ decreases exponentially as $j$ increases. Alice runs the algorithm $T$ on the distribution $D_x$ and would like to send an encoding of the run (i.e., the sampled element along with its probability for each step). Using a trivial encoding, even to send a sampled element, requires $\Theta(\log n)$ bits, which is already more than the Shannon entropy and thus would not give any lower bound. So Alice needs to use a slightly clever encoding. Roughly speaking, since the probabilities are exponentially decreasing, the conditional sampling from any set $S \subseteq [n]$ returns an element from the first "few" smallest elements of $S$ (with high probability). Thus even though $|S|$ can be large, the sampled element (at each step) can be specified by only constantly many bits. Alice sends this encoding of the sampled element (along with its probability value which can again be encoded with constantly many bits due to the construction of $D_x$'s) at each of $t$ steps to Bob. Hence Bob knows the complete run and thus can determine the index $x$ using the algorithm $T$. We provide the detailed proof in Section~\ref{sec:lowerbound}.

Next, we turn our attention to showing the dependency of $\epsilon$ in the lower bound. We show an $\Omega\Big(\frac{1}{\epsilon^2 \log \frac{1}{\epsilon}}\Big)$ lower bound for an additive $\epsilon n$-factor algorithm (\autoref{lem:lbepsilon}). Since a multiplicative $(1+\epsilon)$-factor algorithm also provides an additive $\epsilon n$-factor estimation, the above lower bound also works for $(1+\epsilon)$-factor algorithms. We prove this bound by showing a reduction from a well-studied communication complexity problem, namely the \emph{Gap-Hamming distance} problem (see Section~\ref{sec:additiveLB}), and then applying the known lower bound for the Gap-Hamming distance~\cite{chakrabarti2012optimal}. 
\ksmedits{The proof argument (and the dependence on $\epsilon$) also holds true for {\CD} model, and therefore, the detailed proof is presented in Section~\ref{sec:additiveLB}.

}

\paragraph*{Lower bound for {\CD}. }
The approach used to get an $\Omega(\log \log n)$ lower bound in the {\CE} model cannot be used to show the same for the {\CD} model. One of the powers of {\CD} model (over the {\CE}) comes from its ability to compute $D(S):=\sum_{j \in S} D(j)$ for any set $S \subseteq [n]$. 


Recall the hard instance $D_x$'s used in the $\Omega(\log  \log n)$ lower bound proof for the {\CE} model. Let $X^* := \{2^{x} \mid x \in [\log n]\}$. It is easy to verify that $D_{x}(X^*) \neq D_{x'}(X^*)$ for any $x \neq x' \in [\log n]$. So, the value of $x$ (and hence the support of $D_x$) can be determined using only \ksmedits{one} {\CD} query with the set $X^*$. Thus our lower bound argument fails in this model. For the sake of intuition about how we overcome the above issue, we want to point out that the above argument does not fail if we have an approximate {\CD} query instead of {\CD} query, i.e., if the oracle gives the estimate of $D(S)$ up to a small additive error say $\frac{3}{2^{n^{0.1}}}$. This is because for $x,x' \in [\frac{\log n}{10}, \log n]$, we have $|D_{x}(X^*) - D_{x'}(X^*)| \le \frac{2}{2^{n^{0.1}}}$. Hence, given that $x,x' \in [\frac{\log n}{10}, \log n]$, we can not distinguish between $x$ and $x'$ as the estimate could be the same for both $x$ and $x'$.

To mitigate the above issue (for {\CD}), we construct a new set of hard distributions. Our objective is that for any set $X^*$, if  the value of $D_x(X^*)$ is in $(0,1)$, then value of $D_x(X^*)$ should not give information about $x$. One plausible approach could be to replace each distribution $D_x$ with a finite set of distributions such that for any set $S\subseteq [n]$, there are many distributions in the instance with the same value of $D(S)$. Unfortunately, we do not know how to get such a set of distributions preserving other useful properties needed for our proof. Our key high-level idea is to replace the distribution $D_x$ (for each $x \in [\log n]$) with a distribution over an infinite number of distributions. This way, the value of $D_{x}(X^*)$ cannot be used to determine the value of chosen $x$ (as there can be infinite values of $x$ having the same value $D(X^*)$). However, one immediate issue that arises is if our instance has an infinite domain (here distributions), then how do we even get a distribution over an infinite space? Further, like before, we still want the probabilities to exponentially decrease so that the sampled element is always among the first few smallest elements of the conditioning set $S$. Fortunately, in statistics and compositional data analysis, there has been a study of probability distributions on the set of all (infinite) distributions. Based on requirements, our choice is the well-studied Dirichlet distribution that satisfies a strong independence property (\autoref{lem:independir}) which is necessary for our analysis. 

A Dirichlet distribution on support $[K]$ with parameters $\alpha_1,\cdots,\alpha_K > 0$ has a probability density function given by $f(p_1,\dots,p_K) =  \frac{\prod_{i \in [K]} p_i^{\alpha_i -1}}{B(\alpha_1,\cdots,\alpha_K)}$ where $\{p_i\}_{i \in K}$ belongs to the standard $K-1$ simplex, i.e., $\sum_{i \in [K]} p_i = 1$ and $p_i \ge 0$ for all $i \in [K]$ and $B(\alpha_1,\cdots,\alpha_K)$ is a  normalizing constant. When $\alpha_1 = \cdots =\alpha_K = 1$, the Dirichlet distribution is just the uniform distribution on $K-1$ simplex. The higher the value of parameter $\alpha_i$, the higher the (expected) value of $p_i$. Since we want the probability value to be exponentially decreasing, we set the values of $\alpha_1,\cdots,\alpha_n$ exponentially decreasing. Then for each index $x$ chosen uniformly at random from  $[\log n]$, we sample a distribution $D_x$ with support $[2^x]$ from the Dirichlet distribution with parameters $\alpha_1,\cdots,\alpha_{2^i}$. By the standard Yao's principle, it suffices to show high error probability of any deterministic algorithm that correctly estimates the support size (and hence determines the index $x$) of the distribution sampled as above. Note that the entropy of the index is still $\log \log n$, but the previous communication framework (between Alice and Bob) will not  be useful here. This is because the range of the outcomes of the algorithm (the actual and the conditional probabilities) is continuous and so cannot be encoded with any finite (or even infinite) length message. Instead, we apply Fano's inequality, a tool from information theory. 

Roughly speaking, we show that the information gain (about the index) by the query's outcome at every step is $O(1)$. Since the initial entropy of the index is $\log \log n$, at least $\log \log n$ steps of the algorithm are needed. The main technical challenge is to upper bound the information gain at every step. It is particularly challenging as it requires calculating the explicit density function (for the outcome) corresponding to each index. These density functions are conditioned on the previous outcomes and thus change at every step. Further, the set queried by the algorithms can be adaptive, which makes our task even more difficult. To ease our analysis, we first assume that the queried sets by the algorithm are \emph{laminar}, i.e., either they do not intersect or are subsets/supersets of each other. Our $\Omega(\log \log n)$ lower bound holds for {\CD} model for all the algorithms satisfying this laminar condition. It is not hard to observe that any algorithm that makes $t$ general queries can be simulated by an algorithm that queries the laminar family of sets and makes at most $2^t$ queries. This observation gives us $\Omega(\log \log \log n)$ lower bound for the general case. We believe that $\Omega(\log \log n)$ is the correct lower bound for the general case, but (perhaps it is an artifact of our analysis that) the laminar structure is necessary for applying the independence properties of Dirichlet distribution which leads to only an $\Omega(\log \log \log n)$ lower bound. 
\ksmedits{For the sake of simplicity, we first prove the lower bound for a weaker model called, {\SEval}, (in Section~\ref{sec:seval-LB}), and then extend to the general {\CD} model (in Section~\ref{sec:cd-LB}). We refer to the oracle that, given any $S \subseteq [n]$, just outputs the value of $D(S)$, as {\SEval} oracle. Since using a {\CD} query we can simulate a {\SEval} query, the {\CD} model is at least as powerful as {\SEval}. 
We now describe the proof of the lower bound for the {\SEval} model in more detail.
}

For an index $x$ chosen uniformly at random from  $[\log n]$, we sample a distribution $D_x$ with support $[2^x]$ from the Dirichlet distribution with parameters $\alpha_1,\cdots,\alpha_{2^x}$ ($\alpha_j = \frac{1}{2^j}$ for all $j \in [x]$). By the standard Yao's principle, it suffices for us to show a high error probability of any deterministic algorithm that correctly estimates the support size (and hence determines the index $x$) of the distribution sampled as above. Let $T$ be any such deterministic algorithm that queries a laminar family of sets. We need to show that $T$ must make $\Omega(\log \log n)$ queries. 
The main technical ingredient in the proof is to show that the information gain (about the index $x$) by the outcome of any query of the algorithm is $O(1)$ (\autoref{clm:boundonIstep}). This implies that the total information gain is $O(t)$, where $t$ is the number of queries (\autoref{clm:boundonIstep}). Then Fano's inequality immediately implies that $t \ge \Omega(\log \log n)$.  Let the $i$-th query ($i \in [t]$) be denoted by set $A_i$ and the outcome of the $i$-th query (sum of probabilities of elements in the set $A_i$) be denoted by $Z_i$.
The information gain (about the index $x$) by the $i$-th query (for any $i \in [t]$) is the conditional mutual information $I(X;Z_i|Z^{i-1})$ where $Z_i$ is the random variable denoting the  outcome of $i$-th query, $Z^{i-1} = (Z_1,\dots,Z_{i-1})$ is the random variable denoting the vector of previous outcomes and $X$  is the random variable denoting the uniformly chosen index $x$ from $[\log n]$. 
By definition, the conditional mutual information $I(A;B|C)$ for random variables $A,B,C$ is equal to the expectation (over $C$) of the KL divergence between the joint distribution $Q_{(A,B)|C}$ and the product distribution $Q_{A|C} \times  Q_{B|C}$. The technical difficulty is to upper bound the conditional mutual information  $I(X; Z_i|Z^{i-1})$ by $O(1)$. This is particularly challenging since the joint and the product distributions are not explicitly given, and queries are adaptive. To overcome this difficulty,  at any step $i \in [t]$ (after $(i-1)$-th query), we first partition the $[\log n]$ into four groups denoted by $L^i_0, L^i_1, L^i_2, L^i_3$ such that for the first three groups, the outcome of the algorithm (i.e., $Z_i$) is (deterministically) determined by the previous outcomes of the algorithm (i.e., $Z^{i-1} =(Z_1,\dots,Z^{i-1}))$ whereas for any $x$ in the fourth group $L^i_3$, the outcome is not fixed (given previous outcomes) but comes from a distribution (which we show to be also Dirichlet). Let this distribution be denoted by $Q_x$ for $x \in L^i_3$. 
 The information gain by the $i$-th query (\autoref{upperboundonI}) is $\log 3$ (because there are three groups for which outcome is deterministically determined) plus the information gain corresponding to the last group $L^i_3$. This term can be upper bounded by the maximum KL divergence between distributions $Q_x$ and $Q_{x'}$ for any $x,x' \in L^i_3$. Thus our goal is to show that KL divergence between $Q_x$ and $Q_{x'}$ for any $x,x' \in L^i_3$  is  $O(1)$. Using the independence property of Dirichlet distributions (\autoref{lem:independir}) and the laminar structure of query sets,  we show that the KL divergence between the distributions $Q_x$ and $Q_{x'}$ is equal to the KL divergence between two beta distributions with different parameters (beta distributions are a special case of Dirichlet distributions). The explicit formula for KL divergence between two beta distributions is well-known (e.g., see~\cite{johnson1995continuous}), and we use this formula to upper bound the KL divergence by $O(1)$.
\begin{remark}
Technically, the above lower bound for {\SEval} (and  {\CD}) does not hold if it is promised that the probabilities in the given distribution are rational numbers. This is because, in the lower bound instances above, the probability of an element can be any arbitrary real number in $(0,1)$. However, we can use the following standard fact to show that the lower bound holds even with the rational probabilities.

A Polya urn is an urn containing $\alpha_i$ balls of color $i$, for each $i \in [K]$. The urn evolves at each discrete time step -- a ball is sampled uniformly at random. The ball's color is observed, and two balls of the observed color are returned to the urn. Let $X_{i,m}$ be the number of balls of color $i$ (for each $i \in [K]$) added after $m$ time steps. Clearly $D_m = (\frac{X_{1,m}}{m},\cdots,\frac{X_{K,m}}{m})$ is a probability distribution over $[K]$. It can be shown that the distributions $D_1,\cdots,D_m$ converges to a Dirichlet distribution with parameters $\alpha_1,\dots,\alpha_K$ when $m$ tends to $\infty$~\cite{blackwell1973ferguson}. 

Instead of using Dirichlet distribution in our lower bound proof, we can use the distribution $D_m$ for sufficiently large $m$. Since for any $m$, the probabilities in $D_m$ are rational numbers, we can get the lower bound even when the probabilities are rationals.
\end{remark}

\paragraph*{The power of {\CD}. }


We further demonstrate the power of the {\CD} model by showing an algorithm with constant query complexity for a number of distribution testing problems for which there are strong lower bounds known for the {\CE} and {\CO} model. Our first example is the well-studied \emph{Equivalence testing} problem. Here, given two distributions $D$ and $D'$, the goal is to accept if $D = D'$ and reject if their \emph{total variation distance} $||D-D'||_{TV} = \sum_{i \in [n]}|D(i)-D'(i)| > \epsilon$ (both with high probability). It is known that $\Omega(\sqrt{\log \log n})$ queries are necessary in the {\CO} model~\cite{acharya2015chasm}. On the other hand, $\Omega(1/\epsilon)$ queries are required in the {\CE} model. (Consider a uniform distribution $U$ on $[n]$. Now randomly choose $i,j \in [n]$. We modify $U$ to construct another distribution $U'$ by setting $U'(i) = 2/n$, $U'(j) = 0$, and no changes in the probability mass of other elements. Note that $\epsilon = ||U-U'||_{TV} = 2/n$. It is easy to see that $\Omega(n) = \Omega(\frac{1}{\epsilon})$ queries are required to distinguish $U$ and $U'$ in the {\CE} model.) We show that  Equivalence testing can be done in just two {\CD} queries. The above upper bound result extends to another unrelated problem for the {\CD} model -- the problem of testing if the given distribution is \emph{$m$-grained}, i.e., the probability of each element is an integer multiple of $1/m$. 
Finally, we show that the multiplicative $(1+\epsilon)$-approximation of square of the $L_2$ norm $(\sum_{j \in [n]} D(j)^2)$ of a distribution $D$ can be computed using $O(\frac{1}{\epsilon^2})$ queries of {\CD}. To the best of our knowledge, this problem has been studied previously only in the {\Samp} model~\cite{goldreich2011testing}, wherein it was shown that $\Omega(\frac{\sqrt{n}}{\epsilon^2})$ queries are required. We refer the readers to Section~\ref{sec:power-SEval} for the details of the above upper bound results.

\paragraph*{The power of bounded-set conditioning. }We further study the support size estimation problem when we allow  {\Samp} oracle access and conditioning  on sets of size at most $k$ (see Section~\ref{sec:bounded-cond}). We show a lower bound of $\Omega(n/k)$ (\autoref{thm:boundedsetLB}) and an upper bound of $O(\frac{n \log \log n}{k})$ (\autoref{thm:boundedsetUB}) for constant factor approximation in this model. The upper bound holds for the {\CO} oracle model, while our lower bound holds for the stronger {\CD} oracle model.

Both the upper and lower bounds are not difficult to establish. Falahatgar et al.~\cite{falahatgar2016estimating} showed that $O(\log \log n)$ queries are sufficient (with no restriction on the size of the set for conditioning) to get a constant approximation  for oracle access which, given any $S \subseteq [n]$, returns whether $S \cap \supp(D) = \emptyset$ or not. Oracle access to a set of size $s$ can be simulated by $s/k$ oracle access when conditioning on at most $k$-sized sets is allowed. This gives an upper bound of  $O(\frac{n \log \log n}{k})$. Interestingly, the hard instance for the bounded-set conditioning to estimate the support size is when the support size is constant. We formally prove our lower bound in Section~\ref{sec:bounded-cond}.


\subsection{Conclusion}
\label{sec:conclusion}
We investigate the power of conditioning for estimating the support size up to a  multiplicative $(1+ \epsilon)$-factor. Till date, there is a gap between the upper bound of $O(\log \log n + 1/\epsilon^2)$ and the lower bound of $\Omega(\sqrt{\log \log n})$ in the standard {\CO} model. In this paper, we close this gap by providing a lower bound of $\Omega(\log \log n + \frac{1}{\epsilon^2 \log \frac{1}{\epsilon}})$. We actually show the lower bound in even a more powerful model, namely {\CE}, where in addition to the conditioning, one is also allowed to get the actual probability of the sampled elements (i.e., a combination of {\CO} and {\PRSamp}). In the dependency of $\epsilon$, there is a small gap of $\log (1/\epsilon)$ factor, and we want to leave the problem of removing this factor from the lower bound term as an open problem.

It is quite surprising that the combination of {\CO} and {\PRSamp} does not yield more power compared to only the {\CO} model in the context of the support size estimation. We thus continue our investigation by appending the algorithms with an even more powerful oracle that could also get the conditional probabilities of the sampled elements (not just the actual probabilities). We call this model {\CD}. This model turns out to be more powerful in the context of several other important distribution testing problems, as demonstrated in Section~\ref{sec:power-SEval}. For the support size estimation, we show a lower bound of $\Omega(\log \log \log n + \frac{1}{\epsilon^2 \log \frac{1}{\epsilon}})$ in this {\CD} model. On the technical side, this paper introduces many new ideas, such as using continuous distribution (Dirichlet distribution) for constructing hard instances and applying information theory and communication complexity tools to conditional sampling models. We hope that such techniques could be useful for showing non-trivial lower bounds for other distribution testing problems as well. 

For the support size estimation problem in the {\CD} model, currently, we only know of an $O(\log \log n)$ upper bound, whereas we could only show a lower bound of $\Omega(\log \log \log n)$. We would like to pose the problem of closing this gap as an interesting open problem.

\section{Preliminaries}
\label{sec:prelims}
\paragraph*{Notations. }We use the notation $[n]$ to denote the set of integers $\{1,2,\cdots,n\}$. For any probability distribution $D$ defined over $[n]$, for any $i \in [n]$, let $D(i)$ denote the probability of choosing $i$ when sampling according to $D$. For any subset $S \subseteq [n]$, we use $D(S)$ to denote the probability mass assigned on $S$ by the distribution $D$, i.e., $D(S) := \sum_{i \in S}D(i)$.

\paragraph{Different access models. }
Let $D$ be a distribution over $[n]$. Below we formally define the query models that we consider in this paper.
\begin{definition}[{\CO} Query Model]
A \emph{conditional} (in short, {\CO}) oracle for $D$ takes as input a set $S \subseteq [n]$, and if $D(S) > 0$, returns an element $j \in S$ with probability $D(j)/D(S)$. If $D(S) = 0$, then the oracle returns "failure".
\end{definition}

\begin{definition}[{\CE} Query Model]
A \emph{probability-revealing conditional sampling} (in short, {\CE}) oracle for $D$ takes as input a set $S \subseteq [n]$, and if $D(S) > 0$, returns a pair $(j,D(j))$  (where $j \in S$) with probability $D(j)/D(S)$. If $D(S) = 0$, then the oracle returns "failure".
\end{definition}

\begin{definition}[{\CD} Query Model]
A \emph{conditional evaluation} (in short, {\CD}) oracle for $D$ takes as input a set $S \subseteq [n]$, and if $D(S) > 0$, returns a tuple $(j,D(j),D(j)/D(S))$  (where $j \in S$) with probability $D(j)/D(S)$. If $D(S) = 0$, then the oracle returns "failure".
\end{definition}

\begin{definition}[{\SEval} Query Model]
A \emph{set evaluation} (in short, {\SEval}) oracle for $D$ takes as input a set $S \subseteq [n]$, and returns the value $D(S)$.
\end{definition}

It is straightforward to observe that the {\CD} is at least as powerful as the {\CE} oracle which in turn is at least as powerful as the {\CO} oracle. Further, the {\CD} oracle is at least as powerful as the {\SEval} oracle.\footnote{Since on input $S$, the {\CD} oracle returns a tuple $(j,D(j),D(j)/D(S))$ where $j \in S$, one can compute the value of $D(S)$ whenever $D(S) > 0$; otherwise (when $D(S) = 0$) the {\CD} oracle returns "failure", from which one can infer that $D(S) = 0$.}

\paragraph{Shannon entropy and source coding theorem. }
The \emph{entropy} of a discrete random variable $X$ taking values in $\mathcal{X}$ is defined as $H(X) := \sum_{x \in \mathcal{X}} p(x) \log \frac{1}{p(x)}$ where $p(x) = \Pr[X =x]$. 

The seminal work of Shannon~\cite{Sha48} establishes a connection between the entropy and the expected length of an optimal code that encodes a random variable.
\begin{theorem}[Shannon's Source Coding Theorem~\cite{Sha48}]
\label{thm:source-coding}
Let $X$ be a discrete random variable over domain ${\cal X}$. Then for every uniquely decodable code $C :{\cal X} \to \{0,1\}^*$,  $\mathbb{E}(|C(X)|) \ge H(X)$. Moreover, there exists a uniquely decodable code $C:{\cal X} \to \{0,1\}^*$ such that $\mathbb{E}(|C(X)|) \le H(X)+1$.
\end{theorem}

\section{Lower Bound in {\CE} Model}
\label{sec:lowerbound}
In this section, we show a lower bound of  $\Omega(\log \log n + \frac{1}{\epsilon^2 \log \frac{1}{\epsilon}})$ {\CE} queries, for any  algorithm that estimates the support size within a multiplicative $(1 + \epsilon)$-factor with high probability. Since {\CE} is more powerful than {\CO}, our lower bound  holds for {\CO} as well. Till date, a lower bound of only $\Omega(\sqrt{\log \log n})$ is known for {\CO} queries~\cite{chakraborty2016power}. We want to emphasize that the lower bound construction in~\cite{chakraborty2016power} is quite complicated and their framework cannot give better than $\Omega(\sqrt{\log \log n})$ lower bound. This is because their lower bound is derived from the fact that the number of leaves  of a decision tree  representing a $q$-query tester is  $O(2^{q^2})$ (note, this bound is tight), which must be at least $ \Omega(\log n)$, implying the $\Omega(\sqrt{\log \log n})$ lower bound.


\condevalLB*

We prove the above theorem by first showing that $\Omega(\log \log n)$ queries are necessary to estimate the support size up to a multiplicative $4/3$-factor. 
Next, we show that to achieve an additive $\epsilon n$-approximation, $\Omega(\frac{1}{\epsilon^2 \log (1/\epsilon)})$ queries are required. Since multiplicative $(1 + \epsilon)$-approximation of the support size also implies an $\epsilon n$-additive approximation of the same, we get a lower bound of $\Omega(\frac{1}{\epsilon^2 \log (1/\epsilon)})$ on multiplicative $(1+ \epsilon)$-approximation algorithms as well.
It turns out $\Omega(\frac{1}{\epsilon^2 \log (1/\epsilon)})$ bound holds for the stronger model of  {\CD} query model (and thus also for the {\CE} model). Therefore, we defer the proof of lower bound for $\epsilon$-dependence to Section~\ref{sec:additiveLB} for the sake of clarity of exposition.


\begin{lemma}
\label{lemma:loglogn}
Any algorithm that, given {\CE} access to a distribution $D$ on $[n]$, estimates $|Supp(D)|$ within a multiplicative $\frac{4}{3}$-factor with probability at least $\frac{9}{10}$, must make $\Omega(\log \log n )$ queries to the {\CE} oracle.
\end{lemma}
We devote the rest of this section in proving the above lemma. 
First, consider the following simple \emph{integer-guessing game} between Alice and Bob: Alice chooses an integer $ x \in [\log n]$ uniformly at random. Bob's task is to guess the integer $x$.

Let us consider the randomized one-way communication complexity of the above problem. More specifically, only Alice can send a message to Bob (but Bob cannot send anything to Alice). Alice and Bob also have access to public randomness. Then the question is how many bits Alice needs to send to solve the above problem with probability at least $2/3$. The one-way communication complexity is defined to be the minimum number of bits that is communicated in any (randomized) protocol that solves the integer-guessing game.

By a simple application of Shannon's source coding theorem (\autoref{thm:source-coding}), we show the following.

\begin{claim}
\label{clm:intguess}
The randomized one-way communication complexity of the above integer-guessing game is $\Omega(\log \log n)$.
\end{claim}
\begin{proof}
Consider a one-way randomized protocol $P$ for the above integer-guessing game, where with probability at least $2/3$, Alice sends at most $t$ bits to Bob and Bob succeeds in guessing Alice's chosen integer. Let $X$ be the random variable denoting the integer chosen by Alice. Let the random variable $R$ denote the public randomness shared between Alice and Bob, to be used by the randomized communication protocol $P$. 

Now define the following encoding function $Enc$ to encode $X$ given $R$: If using $R$ the protocol $P$ succeeds (i.e., Bob outputs Alice's chosen integer correctly), then encode the $t$ bits sent by Alice (according to the protocol $P$); else encode $X$ using a trivial encoding (using at most $\log \log n + 1$ bits). Then $\E[|Enc(X)|] \le t + 1/3 \cdot (\log \log n + 1)$ (since the protocol $P$ succeeds with probability at least $2/3$ while Alice sending at most $t$ bits to Bob). Since $X$ and $R$ are independent, the entropy of $X \mid R$ is $H(X \mid R) = H(X) = \log \log n$. Now it follows from Shannon's source coding theorem (\autoref{thm:source-coding}),
\begin{align*}
    &\E[|Enc(X)|] \ge H(X\mid R)\\
    \Rightarrow & t + 1/3 \cdot (\log \log n + 1) \ge \log \log n
\end{align*}
which in turn implies $t \ge \Omega(\log \log n)$.
\end{proof}

Next, we provide a reduction from the above integer-guessing game to the problem of estimating the support size up to $4/3$-multiplicative factor. We start with constructing a hard distribution for the support size estimation. 

For each $i \in [\log n]$, we construct a distribution $D_i$ on $[n]$ with support size $2^i$ as follows: Set
\begin{align*}
   D_i(j) =
\begin{cases}
1/2^j &\text{ if } j\in [2^i-1]\\
1/2^{2^i - 1} &\text{ if } j=2^i\\
0&\text{ otherwise}.
\end{cases} 
\end{align*}

 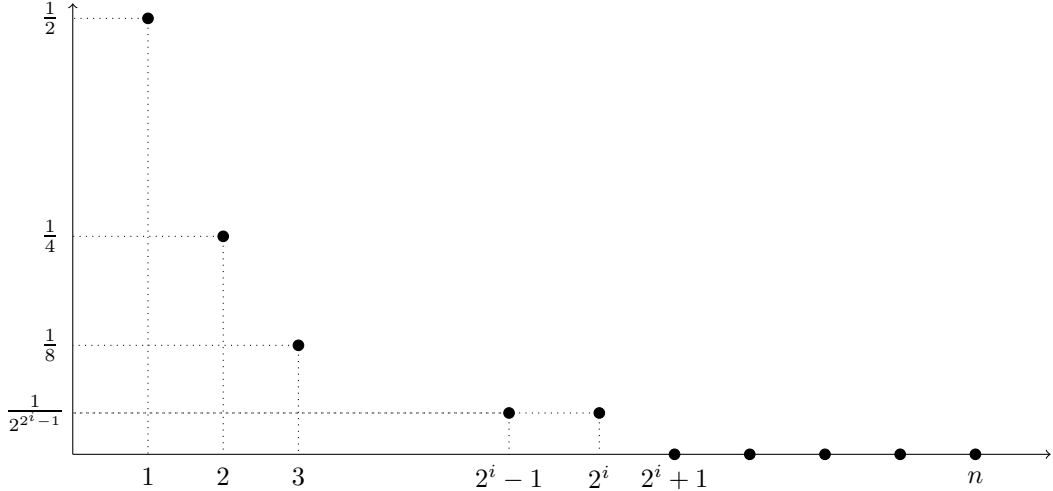
\begin{figure}[!htb]
\begin{tikzpicture}

\draw[->] (-1,0) -- (12,0);
\draw[->] (-1,0) -- (-1,6);

\node[] () at (7,-0.3) {$2^i + 1$};
\node[] (f) at (11,-0.3) {$n$};
\node[]  (g) at (6,-0.3) {$2^i$};
\node[] (h) at (4.8,-0.3) {$2^i -1$};
\node[] (i) at (0,-0.3) {$1$};
\node[] (j) at (1,-0.3) {$2$};
\node[] (k) at (2,-0.3) {$3$};

\node[] (l) at (-1.3,5.8) {$\frac{1}{2}$};
\node[] (m) at (-1.3,2.9) {$\frac{1}{4}$};
\node[] (n) at (-1.3,1.45) {$\frac{1}{8}$};
\node[] (o) at (-1.5,0.55) {$\frac{1}{2^{2^i - 1}}$};

\filldraw[black] (0,5.8) circle (2pt) node[] (a){};
\filldraw[black] (1,2.9) circle (2pt) node[] (b){};
\filldraw[black] (2,1.45) circle (2pt) node[] (c){};
\filldraw[black] (4.8,0.55) circle (2pt) node[](d){};
\filldraw[black] (6,0.55) circle (2pt) node[](e){};
\filldraw[black] (7,0) circle (2pt) node[](){};

\filldraw[black] (8,0) circle (2pt) node[](){};
\filldraw[black] (9,0) circle (2pt) node[](){};
\filldraw[black] (10,0) circle (2pt) node[](){};

\filldraw[black] (11,0) circle (2pt) node[](){};

\draw[dotted] (a) -- (0,0);
\draw[dotted] (b) -- (1,0);
\draw[dotted] (c) -- (2,0);
\draw[dotted] (d) -- (4.8,0);
\draw[dotted] (e) -- (6,0);

\draw[dotted] (a) -- (-1,5.8);
\draw[dotted] (b) -- (-1,2.9);
\draw[dotted] (c) -- (-1,1.45);
\draw[dotted] (d) -- (-1,0.55);
\draw[dotted] (e) -- (-1,0.55);

\end{tikzpicture}
	\caption{The distribution $D_i$}
	\label{fig:distribution}
\end{figure}

Given $S \subseteq [n]$ and $j \in S$, if $j$ is the $k$-th largest element in $S$ then we define rank of $j$ in $S$, denoted by $\rk_S(j)$, to be $k$. E.g., let $s_1<s_2<\cdots<s_c$ be the elements in $S$. Then the rank of $s_r$ (for $r \in [c]$) in $S$ is $\rk_S(s_r)=r$. 

For the sake of contradiction, we assume that there is a randomized algorithm $T$ that, with probability at least $9/10$, estimates the support size of any distribution on $[n]$ within a multiplicative $4/3$-factor by making at most $t=t(n)$ {\CE} queries. Next, we describe a one-way communication protocol between Alice and Bob to solve the integer-guessing game defined previously. First recall, by the construction, $|\supp(D_{i+1})| = 2 |\supp(D_i)|$, for any $i \in [\log n - 1]$. Now, since $T$ estimates the support size of any distribution within $4/3$-factor, from the output of the algorithm $T$, it is straightforward to retrieve $i$ if the input distribution is $D_i$ (for any $i \in [\log n]$) with probability at least $9/10$. Let us now describe a strategy used by Alice and Bob to solve the integer-guessing game using one-way communication.

Alice and Bob share a public random string $R$ (to be used by the randomized algorithm $T$). Recall, in the first step of the integer guessing game, Alice chooses $x \in [\log n]$ uniformly at random. Then Alice wants to construct a message on receiving which Bob can guess $x$ correctly with high probability. To construct the message, Alice considers the distribution $D_x$ (as defined earlier). Then she runs the algorithm $T$ (using the random string $R$) with {\CE} oracle queries to $D_x$. Given $R$, the algorithm $T$ can be viewed as a decision tree $\mathcal{T}_R$ in a straightforward way. Note, each leaf node of this decision tree is labelled by some $i \in [\log n]$. Since $T$ places exactly $t$ queries, the height of the decision tree $\mathcal{T}_R$ is also $t$. At any step $\ell$, let $A_\ell$ denote the set to condition on, and $(z_\ell, p_\ell)$ denote the output of {\CE} oracle after placing a query to $D_x$ condition on the set $A_\ell$. Thus $z_\ell \in A_\ell$ and $p_\ell = D_x(z_\ell)$. Note, at any step $\ell$, the decision tree $\mathcal{T}_R$ chooses the set $A_\ell$ based on $(A_1,z_1,p_1),\cdots,(A_{\ell-1},z_{\ell-1},p_{\ell-1})$. The root to leaf path in $\mathcal{T}_R$ is specified by the sequence $(A_1,z_1,p_1),\cdots,(A_{t},z_{t},p_{t})$. Then Alice constructs a "short" message so that Bob can retrieve the above sequence from her message. 

For an $\ell \in [t]$, let $X_\ell$ denote $\rk_{A_\ell}(z_\ell)$. Instead of directly encoding $z_\ell$, Alice encodes $X_\ell$. Observe, by the construction of the distributions $D_i$'s, for any $\ell\in [t]$, the value of $p_\ell$ can be either $1/2^{z_\ell}$, or $1/2^{z_\ell-1}$, or $0$. Alice uses two bits to encode these three values -- say $00$ for the value $0$, $01$ for the value $1/2^{z_\ell}$, and $10$ for the value $1/2^{z_\ell-1}$. Formally, each pair $(z_\ell,p_\ell)$, Alice uses the encoding $\blk(\ell)$ as follows: First, encode $X_\ell$ using $\lceil \log X_\ell \rceil$ bits. Then append that with two bits to encode the value of $p_\ell$ (as mentioned above). We refer to each $\blk(\ell)$ as \emph{block}. Note, $|\blk(\ell)| = \lceil \log X_\ell \rceil + 2$ .

Ideally, Alice would like to send the string $\blk(1)\circ \blk(2)\circ \cdots \blk(\ell)$ (let $\circ$ denote the concatenation operation). However, note, the length of $\blk(\ell)$ is different for different $\ell$'s. To make sure Bob can decode each encoding block correctly, Alice only uses every alternate bits of her message to encode $\blk(\ell)$'s and sets remaining bits as 1 or 0 depending on whether it is the end of an encoding block or not. More specifically, suppose $\blk(\ell) = b^\ell_1\dots b^\ell_{|\blk(\ell)|}$. Then Alice sends the following message
\[
M_R(x) = b^1_1 \thinspace 0 \thinspace b^1_2  \thinspace 0 \thinspace \dots 0 \thinspace b^1_{|\blk(1)|} \thinspace  1 \thinspace b^2_1 \thinspace 0 \thinspace b^2_2 \thinspace 0  \dots 0 \thinspace b^2_{|\blk(2)|} 1 \cdots  \thinspace  b^t_1 \thinspace 0 \thinspace b^t_2 \thinspace 0 \dots 0 \thinspace b^t_{|\blk(t)|}\thinspace  1
\]
to Bob. Note, an even digit is $0$ unless it marks the end of $\blk(\ell)$ for some $\ell \in [t]$. Total number of bits sent by Alice is 
\begin{equation}
\label{eq:encodinglen}
    |M_R(x)| = 2 \sum_{\ell \in [t]}|\blk(\ell)| = 2 \sum_{\ell \in [t]}(\lceil \log X_j \rceil + 2)  \le 6t + 2 \log (X_1 \cdots X_t).
\end{equation}

Let us now describe how Bob uses the message $M_R(x)$ to decode back the sequence $(A_1,z_1,p_1),\cdots,(A_{t},z_{t},p_{t})$. Recall, the random string $R$ is shared between both Alice and Bob. Thus Bob knows the decision tree $\mathcal{T}_R$ and so the set $A_1$. Now, Bob starts traversing the decision tree using the message $M_R(x)$. He finds the first even bit set to 1 and then discards all the even bits up to that part of the message $M_R(x)$ to retrieve $\blk(1)$. By the construction of $\blk(1)$, the last two bits denotes whether $p_1$ is $0$, or $1/2^{z_1}$, or $1/2^{z_1-1}$. The remaining bits represent $\rk_{A_1}(z_1)$. Once Bob gets the rank of $z_1$ in $A_1$, he retrieves $z_1$ (since he knows the set $A_1$), and then he retrieves the value of $p_1$. Using $(z_1,p_1)$, he traverses the decision tree and reaches the next internal node and gets to know the set $A_2$. Then he extracts $\blk(2)$ from the message $M_R(x)$ as before and proceeds. At the end, he decodes back the whole sequence $(A_1,z_1,p_1),\cdots,(A_{t},z_{t},p_{t})$, and thus reaches the desired leaf node in the decision tree $\mathcal{T}_R$. This completes the description of the strategy used by Alice and Bob.

Now, if the algorithm $T$ using the random string $R$ outputs an estimate of $|\supp(D_x)|$ up to $4/3$-factor, Bob also gets to know that estimated value by following the above described one-way protocol. Observe, since by the construction, $|\supp(D_{i+1})| = 2 |\supp(D_i)|$, for any $i \in [\log n - 1]$, from the estimated value Bob can easily compute $x$ correctly. Hence, the above one-way communication protocol solves the integer-guessing game with probability at least $9/10$.
 
 It only remains to argue about the communication cost (the message length sent by Alice) of the above one-way protocol. For that purpose, we claim the following.
\begin{claim}
\label{clm:rankbound}
With probability at least $5/6$, $X_1X_2\cdots X_t \le 23^t$.
\end{claim}
We defer the proof of the above claim to the end. Let us now provide an upper bound on the message length by assuming the above claim. It implies from~\autoref{eq:encodinglen} that with probability at least $5/6$, $|M_R(x)| \le 16 t$. Hence, with probability at least $5/6 \cdot 9/10 = 3/4$, using the above one-way communication protocol Alice needs to send at most $16 t$ bits and Bob outputs Alice's chosen integer $x$ correctly. It now follows from~\autoref{clm:intguess}, $t \ge \Omega(\log \log n)$ which concludes the proof of~\autoref{lemma:loglogn}.

Now it only remains to prove~\autoref{clm:rankbound}.
\begin{proof}[Proof of~\autoref{clm:rankbound}]
Recall, for each $\ell \in [t]$, $X_\ell$ is a random variable denoting $\rk_{A_\ell}(z_\ell)$. We now bound its expectation and variance. Further recall, $z_\ell$ was sampled from the distribution $D_x$ condition on the set $A_\ell \subseteq [n]$. Suppose $|A_\ell| = m$. Let $a_1 < a_2 <\cdots a_m$ be the elements of $A_\ell$ in the order. Let $q_j$ denote the conditional probability of $a_j$. Let us first provide an upper bound on $q_j$ for any $3 \le j \le m$. Observe, by the construction, $D_x(a_j) \ge 2 D_x(a_{j+1})$ for all $j\in [m-2]$, and $D_x(a_{m-1}) \ge D_x(a_m)$. Thus, we have $D_x(a_j) \le \frac{D_x(a_1)}{2^{j-2}}$ for all $3 \le j \le m$. So we get $q_j = \frac{D_x(a_j)}{D_x(A_\ell)} = \frac{D_x(a_j)}{\sum_{k \in [m]}D_x(a_k)} \le \frac{1}{2^{j-2}}$, for any $3 \le j \le m$.

Now, we have $\E[X_\ell] = \sum_{j\in [m]} j \cdot q_j \le 
q_1 + q_2 + \sum_{j \ge 3} j \cdot \frac{1}{2^{j-2}} < 6$. Also $\E[X_\ell^2] = \sum_{j \in [m]}j^2 \cdot q_j \le 1 + 4 q_2 + \sum_{j\ge 3}\frac{j^2}{2^{j-2}}  \le 23$. 

Consider the random variable $Z:=X_1X_2\cdots X_t$. Since $X_1,\cdots,X_t$ are independent, we have $\E[Z] = \E[X_1]\cdots \E[X_t] \le 6^t$ and $Var(Z) = \E[X_1^2] \cdots \E[X_t^2]  - (\E[X_1])^2 \cdots (\E[X_t])^2 \le 23^t$. By Chebyshev inequality, for any $r$, we have $\Pr[Z \le \E[Z] + r] \ge 1 - Var(Z)/r^2 \ge  1-23^t/r^2$. Set $r = \sqrt{6\cdot23^t}$. Then we get $Z \le 6^t+\sqrt{6\cdot 23^t} < 23^t$ with probability at least $5/6$.
\end{proof}


\section{Lower bound for {\CD}}\label{sec:L1lowerbound}
In this section, we provide a lower bound on the query complexity of the support size estimation problem in the {\CD} query model. Recall, for a given distribution $D$, a {\CD} query with a set $S \subseteq [n]$ outputs an element $j \in S$ sampled as per its conditional probability $D(j)/D(S)$ (like {\CO}), its actual probability $D(j)$ (like {\CE}) and its conditional probability $D(j)/D(S)$. It is easy to see that the ${\CD}$ query model can simulate the {\SEval} model as $D(S)=\sum_{i \in S}D(i)$ can be determined from the values $D(j)/D(S)$ and $D(j)$, for any $S \subseteq [n]$. 

We now illustrate the power of the {\SEval} queries in our case of support size estimation. Recall, in the proof of~\autoref{lemma:loglogn}, we construct, for each $x \in [\log n]$, a distribution $D_x$ on $[n]$ with support size $2^x$. Suppose $x$ is uniformly chosen at random from $[\log n]$. Then we prove any deterministic algorithm that, given {\CE}  access to $D_x$ and estimates the support size of $D_x$ within multiplicative $4/3$-factor with high probability, must make $\Omega(\log \log n)$ queries. We now argue that if the algorithm is given {\SEval} access to $D_x$ then just one query is sufficient to estimate the support size with high probability. Let $X^* := \{2^{x}|x \in [\log n]\}$. It is easy to verify that $D_{x}(X^*) \neq D_{x'}(X^*)$ for any $x \neq x' \in [\log n]$. Hence, the value of $x$ (and hence the support of $D_x)$ can be determined via a {\SEval} query with the set $X^*$.

To mitigate the above issue, we need to construct a new set of hard distributions to show the lower bound. Our idea to construct hard distribution over distributions, for this model, is to replace the distribution $D_x$ (for each $x \in [\log n]$) with a distribution over an infinite number of distributions. This way, the value of $\sum_{j \in X^*}D_{x}(j)$ cannot be used to determine the value of Alice's chosen $x$. We will give a formal description of our hard distribution later. In this section, we first prove the following theorem.

\begin{restatable}{theorem}{cdlowerbound}
\label{thm:CDlowerbound}
Any algorithm that, given {\CD} access to a distribution $D$ on $[n]$ approximates the support size $|\supp(D)|$ within multiplicative $4/3$-factor with probability at least $2/3$, must make $\Omega(\log \log \log n)$  queries.
\end{restatable}

Later, we will show in~\autoref{lem:lbepsilon} that $\Omega(\frac{1}{\log (1/\epsilon) \epsilon^2})$ queries in the {\CD} model are needed to approximate the support size within an additive $\epsilon n$-factor. Together, we get the following lower bound.

\celowerbound*

We first prove a lower bound of $\Omega(\log \log \log n)$ for the weaker {\SEval} model, which illustrates the main idea. Later in Section~\ref{sec:cd-LB} we will generalize this proof idea to the {\CD} model.

\subsection{Warm-up with Lower bound for {\SEval}}
\label{sec:seval-LB}

\begin{theorem}
\label{thm:L1lowerbound}
Any  algorithm that, given {\SEval} access to a distribution $D$ on $[n]$ approximates the support size $|\supp(D)|$ within multiplicative $4/3$-factor with  probability at least $2/3$, must make $\Omega(\log \log \log n)$ queries.
\end{theorem}
To prove the above theorem, we first assume that the queried sets by the algorithm are laminar, i.e., for any two queried sets $A_i,A_j$ either $A_i \cap A_j = \emptyset$ or $A_i \subset A_j$ or $A_j \subset A_i$. We argue that an algorithm that makes $t$ {\SEval} queries can be simulated by an algorithm that makes at most $2^t$ {\SEval} queries such that the queried sets form a laminar family. Next, we show a lower bound of $\Omega(\log \log n)$ on the number of queries with the laminar family assumption. As a consequence, we deduce $\Omega(\log \log \log n)$ query lower bound for algorithms with arbitrary (i.e., with no assumption of queried sets forming a laminar family) {\SEval} queries. 
\begin{lemma}
\label{lem:laminar}
Any algorithm that makes $t$ {\SEval} queries to a given distribution $D$, can be simulated by an algorithm that makes at most $2^t$ {\SEval} queries with the property that  the queried sets are laminar.
\end{lemma}
\begin{proof}
Let algorithm $\mathcal{A}$ makes $t$ queries of {\SEval}. We give an algorithm $\mathcal{B}$ that can simulate $\mathcal{A}$ such that the queried sets by $\mathcal{B}$ are laminar. Let the queries of $\mathcal{A}$, in order, be denoted by $A_1,\dots, A_t$. Then the first set queried by $\mathcal{B}$ is also $A_1$. On receiving the value $D(A_1)$, the algorithm $\mathcal{B}$ (using $\mathcal{A}$) determines the set $A_2$. Then the algorithm $\mathcal{B}$, instead of $A_2$, queries the sets $A_2 \cap A_1$ and $A_2 \cap A^c_1$ (where $A^c_2 =[n]\setminus A_2$), and receives the values $D(A_2 \cap A_1)$ and $D(A_2 \cap A^c_1)$. Note, $D(A_2) = D(A_2 \cap A_1)+D(A_2 \cap A^c_1)$. Thus $\mathcal{B}$ can then determine the set $A_3$ (again using $\mathcal{A}$).

In general, instead of $A_i$, the algorithm $\mathcal{B}$ queries all the sets in 
\[
\mathcal{C} := \{A_i \cap C_1 \cap C_2 \cdots \cap C_{i-1} \mid C_j \in \{A_j,A^c_j\}, \; \forall j \in [i-1]\}.
\]
where $A^c_j = [n]\setminus A^j$. Note that $|\mathcal{C}| = 2^{i-1}$ and $D(A_i) = \sum_{C \in \mathcal{C}} D(C)$. From the value of $D(A_i)$, the algorithm $\mathcal{B}$, using $\mathcal{A}$, determines $A_{i+1}$, and proceeds.

It is easy to see that the sets queried by $\mathcal{B}$ form a laminar family and total number of queries is at most  $\le \sum_{j \in [t]}2^{j-1} \le 2^t$.

\end{proof}

So now to prove~\autoref{thm:L1lowerbound} it suffices to show the following.
\begin{lemma}
\label{lem:L1lowerboundlaminar}
Any  algorithm that queries laminar family of sets and given {\SEval} access to a distribution $D$ on $[n]$, approximates the support size $ |\supp(D)|$ within multiplicative $4/3$-factor with  probability at least $2/3$, must make $\Omega(\log \log n)$ queries.
\end{lemma}

To establish our result we need to introduce a few definitions and known results.

\paragraph*{Basics of information theory. }
The \emph{Kullback–Leibler divergence} or simply \emph{KL divergence} (also called {relative entropy}) between two  discrete probability distributions $P$ and $Q$ defined on same probability space $\mathcal{X}$ is given by :
\[
KL(P||Q) := \sum_{x \in \mathcal{X}}p(x) \log \frac{p(x)}{q(x)}
\]
where $p$ and $q$ are probability mass  functions of $P$ and $Q$ respectively. If $P$ and $Q$ are continuous distributions then the summation is replaced by integration: $KL(P||Q) = \int_{x \in \mathcal{X}} p(x) \log \frac{p(x)}{q(x)} \,dx$. 

Let $X$ and $Y$ be two random variables over the space $\mathcal{X} \times \mathcal{Y}$. If their joint distribution is $Q_{X,Y}$ and marginal distributions $Q_X$ and $Q_Y$ respectively, then the mutual information $I(X;Y)$ is defined as:
\[
I(X;Y) := KL(Q_{X,Y}||Q_X \times Q_Y).
\]
For three random variables $X,Y,Z$, the \emph{conditional mutual information} $I(X;Y|Z)$ is defined as 
\begin{align*}
I(X;Y|Z) :=  \mathbb{E}_Z[KL(Q_{(X,Y)|Z}||Q_{X|Z} \times  Q_{Y|Z})]. 
\end{align*}
If $Z$ is a discrete random variable taking values in $\mathcal{Z}$ then we have 
\begin{align*}
&\mathbb{E}_Z[KL(Q_{(X,Y)|Z}||Q_{X|Z} \times  Q_{Y|Z})] = \sum_{z \in \mathcal{Z}} Q_Z(z) \cdot  KL(Q_{(X,Y)|Z=z}||Q_{X|Z=z} \times Q_{Y|Z=z})\\
& = \sum_{z \in \mathcal{Z}} Q_Z(z) \cdot  I(X;Y|Z=z).
\end{align*}
If $Z$ is a continuous random variable with density function $p$ then 
\begin{align*}
& \mathbb{E}_Z[KL(Q_{(X,Y)|Z}||Q_{X|Z} \times  P_{Y|Z})] = \int_{z \in \mathcal{Z}} Q_Z(z) \cdot KL(P_{(X,Y)|Z=z}||P_{X|Z=z} \times P_{Y|Z=z}) \,dz.\\
& = \int_{z \in \mathcal{Z}} Q_Z(z) \cdot I(X,Y|Z=z) \,dz.
\end{align*}
For any three random variables $X,Y,Z$, the \emph{chain rule} for mutual information says that
\[
I(X;(Y,Z)) = I(X;Y)+I(X;Z|Y).
\]

 \begin{theorem}[Fano's inequality]
 \label{thm:fano}
  Consider discrete random variables $X$ and $\hat{X}$ both taking values in $\mathcal{V}$. Then 
  \[
  \Pr[\hat{X} \neq X] \ge 1 - \frac{I(X;\hat{X})+\log 2}{\log |\mathcal{V}|}.
  \]
 \end{theorem}
 
 Consider the random variables $X,Z,\hat{X}$. 
 If the random variable $\hat{X}$ depends only on $Z$ and is conditionally independent on $X$, then we have $I(X;\hat{X}) \le I(X;Z)$. This inequality is known as the \emph{data processing inequality}. For further exposition, readers may refer to any standard textbook on information theory (e.g.,~\cite{CT06}).

\paragraph*{Dirichlet Distribution. }Now, we describe Dirichlet distribution, which is one of the main building blocks of our hard distribution. For any $K \in \mathbb{N}$, a random variable $P = (P_1,\dots,P_K)$ is said to have a \emph{Dirichlet distribution} with parameters $\alpha_1,\dots,\alpha_K >0$ (denoted by $Dir(\alpha_1,\dots,\alpha_K))$ if it has a probability density function $f(p_1,\dots,p_K)$ given by:
\begin{align*}
   f(p_1,\dots,p_K) =
\begin{cases}
\frac{\prod_{i \in [K]} p_i^{\alpha_i -1}}{B(\alpha_1,\dots,\alpha_K)} &\text{ if } \sum_{i \in [K]} p_i =1, p_i \in (0,1]\\
0 &\text{ otherwise}.
\end{cases} 
\end{align*}

 The denominator $B(\alpha_1,\dots,\alpha_K)$ is a normalization constant  to ensure total probability is $1$. When $\alpha_1 = \dots = \alpha_K  = 1$, the Dirichlet distribution is just a uniform distribution on the set of all probability distributions on $[K]$, i.e., on $\{(p_1,\dots,p_K)|\sum_{i=1}^{K}p_i =1,\;\forall_{i \in [K]} p_i \ge 0\}$. It is known that  $\E[P_i] = \frac{\alpha_i}{\sum_{i\in [K]}\alpha_i}$. 
Thus higher the value of $\alpha_i$, the higher is the expected value of $P_i$. For the special case when $K=2$, the Dirichlet distribution is known as beta distribution denoted as $Beta(\alpha_1,\alpha_2)$. 

\textbf{Independence property of Dirichlet distribution:} Our proof will crucially rely on the independence property of Dirichlet distribution stated below. 
\begin{proposition}
\label{lem:independir}
Fix $A \subseteq [K]$, $B \subset A$ and $d \in (0,1)$. Suppose $(P_1,\dots,P_K) \sim Dir(\alpha_1,\dots,\alpha_K)$.  Then  conditioned on $\sum_{i \in A} P_i = d$, we have $\frac{\sum_{i \in B}P_i}{d} \sim Beta(\sum_{i \in B} \alpha_i,\sum_{i \in A\setminus B}\alpha_i)$.
\end{proposition}

\begin{remark}
If $A \cap B \neq \emptyset$ and neither $A \subset B$ nor $B \subset A$, then there is no explicit distribution of $\frac{\sum_{i \in B}P_i}{d}$ given $\sum_{i \in A} P_i = d$. This is why, to prove our lower bound, we assume that the sets queried by the algorithm form a laminar family.
\end{remark}




\paragraph*{Yao's Principle. }Yao's principle~\cite{yao1977probabilistic} is a standard tool to show lower bounds for randomized algorithms. It relates two kinds of randomness -- the randomness over inputs and the randomness used by the randomized algorithm. 

Consider a problem over inputs $\mathcal{X}$ (the set $\mathcal{X}$ can be infinite). Let $\gamma$ be a probability distribution over $\mathcal{X}$ and $X\in \mathcal{X}$ be an input chosen as per $\gamma$.  Let $\mathcal{T}$ be the set of all deterministic algorithms with complexity $t$ (for us, it is query complexity). For $x \in \mathcal{X}$ and $T \in \mathcal{T}$, let the cost $c(x,T)$ take the value $1$ if $T$ gives the wrong answer on input $x$ and $0$ otherwise. Therefore, $\Pr[T \thinspace \text{gives wrong answer on} \thinspace x] = \E[c(x,T)]$. 

Note that a randomized algorithm $R$ is just a probability distribution over the set of deterministic algorithms $\mathcal{T}$. Let $r$ be any probability distribution over $\mathcal{T}$ and $R$ be a random algorithm chosen as per $r$.

From Yao's principle, we have
\begin{align*}
& \max_{x \in \mathcal{X}} \Pr[R \thinspace \text{gives wrong answer on} \thinspace  x] = \max_{x \in \mathcal{X}}\E[c(x,R)] \\
& \ge \min_{T \in \mathcal{T}}\E[c(X,T)] = \min_{T \in \mathcal{T}} \Pr[ T \thinspace \text{gives wrong answer on} \thinspace  X].
\end{align*}

\paragraph*{Proof of~\autoref{lem:L1lowerboundlaminar}. }From Yao's principle, it suffices to consider a distribution $\gamma$ over a set of inputs $\mathcal{X}$ (note, in our case, an input is a distribution) and prove the statement of~\autoref{lem:L1lowerboundlaminar} for any deterministic algorithm with the distribution $D$ being sampled from $\gamma$. Now we define a set of distributions $\mathcal{X}$ on $[n]$ and a distribution $\gamma$ over it.

For any $x \in [\log n]$, $\mathcal{X}_x$ is a set of all probability distributions on $[n]$ with the support $[2^x]$, i.e., 
\[
\mathcal{X}_x := \{(p_1,\dots,p_n)|\sum_{i =1}^{2^x} p_i = 1, \text{ and } \forall_{i \in [n]}, p_i \ge 0, \text{ and } \forall_{i \in [2^x]}, p_i > 0\}.
\]
Let $\mathcal{X} := \cup_{x \in [\log n]}\mathcal{X}_x$. Now we define the distribution $\gamma$. 

Let $\alpha_j = \frac{1}{2^j}$ for all $j \in [n]$. The distribution $\gamma$ first picks $X \in [\log n]$ uniformly at random. Then it samples  $(P_1,\dots,P_{2^X}) \sim Dir(\alpha_1,\dots,\alpha_{2^X})$ and returns the distribution $(P_1,\dots,P_{2^X},0,\dots,0)$. We use $D_X$ to denote the distribution $(P_1,\dots,P_{2^X},0,\dots,0)$. 


  Let $T$ be any deterministic algorithm that estimates the support size (up to $4/3$-factor) of the distribution $D_X$ sampled from $\gamma$ with probability at least $2/3$ (where the probability is over the random choice of input distributions). Let us denote the output of $T$ by $N$. 
   Let $\hat{X} \in [\log n]$ be such that $2^{\hat{X}} \in [\frac{3}{4} N, \frac{4}{3} N]$. Since $|\supp(D_{x+1})| = 2 |\supp(D_x)|$, for every $x \in [\log n]$) and $N$ is within a $4/3$-multiplicative factor of the support size, $\hat{X} = X$. Therefore, 
   \begin{equation}
       \label{eq:estimatehigh}
       \Pr[X = \hat{X}] \ge \Pr\Big[\frac{3}{4} |\supp(D_X)| \le  N \le \frac{4}{3} |\supp (D_X)|\Big] \ge 2/3.
   \end{equation}
   For simplicity, from now on, we consider $\hat{X}$ as the output of the algorithm.
  
  
  


Note that both $X$ and $\hat{X}$ are random variables taking values in $[\log n]$. By Fano's inequality (~\autoref{thm:fano}), we have 
\begin{equation}
       \label{eq:applying-fano}
       \Pr[X = \hat{X}] \le \frac{I(X;\hat{X})+\log 2}{\log \log n}.
   \end{equation}

Let $t$ be the number of the queries made by $T$ (as mentioned in statement of~\autoref{lem:L1lowerboundlaminar}, $T$ has access to {\SEval} oracle and queries laminar family of sets). We will show $I(X;\hat{X}) = O(t)$. Hence, for $\Pr[\hat{X} = X]$ to be at least $2/3$, the number of queries $t$ must be $\Omega(\log \log n)$.



 Let the $i$th query ($1 \le i \le t$) of $T$ be denoted by $A_i$. Note that the set $A_1$ is fixed as $T$ is deterministic (whereas $A_2,\dots,A_t$ are random variables). For any set $A \subseteq [n]$,  $D_X(A) = \sum_{i \in A} D_X(i)$ is the total probability mass of set $A$ in the (randomly chosen distribution) $D_X$.
 Let $Z_r$ ($r \in [t]$) be the random variable equal to  $D_X(A_r)$.  
 At any step, the algorithm determines the set $A_i$ based on the previous outcomes $Z_1,\dots,Z_{i-1}$.
 \begin{claim}
 \label{clm:boundonIstep}
 If for all $i \in [t]$, $I(X;Z_i|Z_1,\dots,Z_{i-1}) = O(1)$, then $I(X;\hat{X}) = O(t)$.
 \end{claim}
 \begin{proof}
 The output $\hat{X}$ is  determined by the values of $Z_1,\dots,Z_t$. Thus, 
\begin{align*}
    I(X;\hat{X}) &\le I(X;Z_1,\dots,Z_t)&&\text{(by the data processing inequality)}\\
    &= \sum_{i \in [t]}I(X;Z_i|Z_1,\dots,Z_{i-1})&&\text{(by the chain rule of mutual information)}.
\end{align*}
 Hence, if for all $i \in [t]$, $I(X;Z_i|Z_1,\dots,Z_{i-1}) = O(1)$, then $I(X;\hat{X}) = O(t)$. 
 \end{proof}
 
 Next, we claim the following.
 \begin{lemma}
 \label{lem:boundonIstep}
 For all $i \in [t]$, $I(X;Z_i|Z_1,\dots,Z_{i-1}) = O(1)$.
 \end{lemma}
  Before proving the above, let us argue that the above implies~\autoref{lem:L1lowerboundlaminar}. It follows from~\autoref{eq:applying-fano} along with~\autoref{clm:boundonIstep}, that
 \[
 \Pr[\hat{X} =X] \le O(\frac{t}{\log \log n}).
 \]
 Then by~\autoref{eq:estimatehigh}, $t \ge \Omega(\log \log n)$. This completes the proof of~\autoref{lem:L1lowerboundlaminar}.
 
 So it only remains to prove~\autoref{lem:boundonIstep}.


\paragraph*{Proof of~\autoref{lem:boundonIstep}. }We use $Z^{i-1}$ to denote the tuple $(Z_1,\cdots,Z_{i-1})$. For any random variable $Y$, we use $Q_Y$  to denote the corresponding probability distribution. When clear from the context, we will drop $Y$ from the subscript. From the definition of conditional mutual information, we have 
\begin{align}
\label{eq:I-Jbound}
I(X;Z_i|Z_1,\dots,Z_{i-1}) & =  \int_{z^{i-1}} Q_{Z^{i-1}}(z^{i-1}) I(X;Z_i|Z^{i-1} = z^{i-1}) \,dz^{i-1}
\end{align}


The following claim is now immediate from~\autoref{eq:I-Jbound}.
\begin{claim}
\label{clm:Igivenz}
If for every $i \in [t]$ and $z^{i-1}$, $I(X;Z_i|Z^{i-1} = z^{i-1}) =  O(1)$, then $I(X;Z_i|Z_1,\dots,Z_{i-1}) = O(1)$.
\end{claim}

\begin{lemma}[\cite{scarlett2019introductory}]
    \label{lem:mutual-inf-covering}
    Consider two random variables $X$ and $Z$, where $X$ is discrete. Suppose there exists $N$ distributions $R_1(z),\dots,R_N(z)$ such that for all $x \in X$ and for some $m > 0$, it holds that
    \begin{align*}
        \min_{j \in [N]} KL(Q_{Z|X}(\cdot|x)||R_j) \le m.
    \end{align*}
    Then we have 
    \begin{align*}
        I(X;Z) \le \log N +m.
    \end{align*}
\end{lemma}
The above lemma is informally called mutual information bound via covering. Suppose  we can partition $X$ into $N$ sets $L_1,\dots,L_N$ such that for all $x \in L_j$, the distributions $\{Q_{Z|X =x}(\cdot|x)\}_{x \in L_j}$ are close (note $Q_{Z|X =x_1}(\cdot|x_1)$ and $Q_{Z|X =x_2}(\cdot|x_2)$ can be different if $x_1,x_2$ are not in same $L_j)$ then maximum information that $Z$ can give about $X$ is $\log N$. The above lemma formalizes this intuition.

Consider any $i \in [t]$. We now fix $z^{i-1}=(z_1,\cdots,z_{i-1})$. We will now partition $X$ into $L^i_{-1},L^i_0,L^i_1,L^i_2, L^i_3$ such that the distributions $\{Q_{Z|X =x}(\cdot|x)\}_{x \in L^i_j}$ will be close. 

Let $L^i$ be the set of all $x \in [\log n]$ such that there exists a distribution $D \in \mathcal{X}_x$ consistent with the values $z^{i-1}$, i.e., $\sum_{j \in A_k}D(j) = z_k$ for all $k \in [i-1]$. So, for any $x \not \in L^i$, we have $Q_{X|Z^{i-1} = z^{i-1}}(x|z^{i-1}) = 0$. We have $L^i_{-1} := [\log n] \setminus L^i$.
 
We will show a partition of the set $L^i$ into the subsets $L^i_0,L^i_1,L^i_2$ and $L^i_3$ such that
 \begin{itemize}
     \item For each $x \in L^i_0$, the value of $z_i$ will always be equal to $0$ (to apply Lemma~\ref{lem:mutual-inf-covering}, we will set $R^i_0$ as a distribution that takes value $0$ with probability $1$). 
     \item For each $x \in L^i_1$, the value of $z_i$ will always be equal to $1$ (set $R^i_1$ as a distribution that takes value $1$ with probability $1$).
     \item For each $x \in L^i_2$, the value of $z_i$ will (be $\ne 0,1$ and) always be completely determined by the value of $z^{i-1}$ (i.e., a fixed number depending on $z^{i-1}$) (set $R^i_2$ as a distribution that takes this fixed value with probability $1$).
     \item $ L^i_3 := L^i \setminus (L^i_0 \cup L^i_1 \cup L^i_2)$ (set $R^i_3$ as the distribution $Q_{Z|X=x}(\cdot|x)$ for any $x \in L^i_3$). 
 \end{itemize}

 This partition directly implies the following lemma.
 \begin{lemma}
 \label{upperboundonI}
 For every $i \in [t]$, $z^{i-1}$, and $x \in [\log n]$, 
 \[
 I(X;Z_i|Z^{i-1}=z^{i-1})  \le \log 4 +  \max_{x,x' \in L^i_3} KL(Q_{Z_i|X=x,Z^{i-1}= z^{i-1}} || Q_{Z_i|X=x', Z^{i-1}=z^{i-1}}).
 \]
\end{lemma}
\begin{proof}
    Conditioned on $Z^{i-1}=z^{i-1}$, we apply Lemma~\ref{lem:mutual-inf-covering} by substituting $Z_i$ for $Z$,  distributions  $R^i_0, R^i_1, R^i_2, R^i_3$ for  $R^1, R^2, R^3, R^4$ respectively and $N = 4$.
\end{proof}

 
 Let us now define the partition of $L^i$. Consider the $i$-th set $A_i$ queried (on the {\SEval} oracle) by the algorithm.
 
 \paragraph*{Partition $L^i_0$:} Let $L^i_0=\{x \in L^i: [2^x] \cap A_i = \emptyset\}$. So $z_i$ must be $0$ for all $x \in L^i_0$. 
 
 \paragraph*{Partition $L^i_1$:} Let $L^i_1 = \{x \in L^i: [2^x] \subseteq A_i\}$. So $z_i$ must be $1$ for all $x \in L^i_1$. 
 
\paragraph*{Partition $L^i_2$:} For any $1\le j \le i$, let $\chi^x(A_j) \in \{0,1\}^{2^x}$ be the characteristic vector of $A_j$ restricted to $[2^x]$. 
 
 Let $x_{\ell}$ be the largest $x \in L^i$ such that $ \chi^x(A_i) = \sum_{j\in [i-1]} \beta_j \cdot \chi^x(A_{j})$ for some reals $\beta_1,\dots,\beta_{i-1}$. Since $\chi^{x_{\ell}}(A_i) = \sum_{j\in [i-1]} \beta_j \cdot \chi^{x_{\ell}}(A_{j})$, we have that for all $x \in L^i$ such that $x \le x_{\ell}$, the following holds
 \[
\chi^{x}(A_i) = \sum_{j\in [i-1]} \beta_j \cdot \chi^{x}(A_{j}).
 \]
 Observe, $\chi^{x}(A_i) = \sum_{j\in [i-1]} \beta_j \cdot  \chi^{x}(A_{j})$ implies $z_i = \sum_{j \in [i-1]}\beta_j \cdot  z_j$ for all $x \in L^i$ such that $x \le x_{\ell}$. 
 
 Next, define $L^i_2 : = \{x\in L^i \mid x \le x_{\ell}\} \setminus (L^i_0 \cup L^i_1)$. 
 \paragraph*{Partition $L^i_3$:} $L^i_3:=  L^i\setminus(L^i_0 \cup L^i_1 \cup L^i_2)$.
 
 We will now give  expression for the distribution $Q_{Z_i|X=x, Z^{i-1}=z^{i-1}}$ for $x \in L^i_3$. We will provide the explicit expression for $KL(Q_{Z_i|X=x,Z^{i-1}= z^{i-1}} || Q_{Z_i|X=x', Z^{i-1}=z^{i-1}})$ later (\autoref{eq:KL-bound}).  Recall, ultimately, we want to show 
 
 \[
 \max_{x,x' \in L^i_3} KL(Q_{Z_i|X=x,Z^{i-1}= z^{i-1}} || Q_{Z_i|X=x', Z^{i-1}=z^{i-1}}) = O(1).
 \]
 



For any $x \in L^i_3$, let $A^x_j := A_j \cap [2^x]$ ($1 \le j \le i)$ be the restriction of $A_j$ to $[2^x]$. Recall, we assume $A_i$'s are laminar. Thus it is straightforward to make the following observation.
\begin{observation}
\label{obs:laminar}
 For any $u,v \in [t]$, if $A_u \subseteq A_v$ then $A^s_u \subseteq A^s_v$. 
\end{observation}
However, it is possible that $A_u \subset A_v$, but $A^s_u = A^s_v$. Let $x_m $ be the maximum value of $x \in L^i_3$.

\begin{claim}
\label{clm:subsetrelation}
For any $1 \le u,v \le i$ and $x' \in L^i_3$, $A^{x_m}_u \subset A^{x_m}_v$ if and only if $A^{x'}_u \subset A^{x'}_v$ .
\end{claim}
\begin{proof}
Consider any $x' \in \Bar{L}$. Clearly, if $A^{x'}_u \subset A^{x'}_v$, then $A^{x_m}_u \subset A^{x_m}_v$. So it only remains to show the converse.

Consider any $1 \le u,v < i$ and let $A^{x_m}_u \subset A^{x_m}_v$. Then if $A^{x'}_u$ is not a strict subset of $A^{x'}_v$. then we must have $A^{x'}_u = A^{x'}_v$. Since $x_m,x' \in L^i_3 \subseteq L^i$, by the definition, there must exist some distributions $D_{x'} \in \mathcal{X}_{x'}$ and $D_{x_m} \in \mathcal{X}_{x_m}$ consistent with the values $z_1,\dots,z_{i-1}$. Note, in $D_{X=x'}$, we have 
\[
z_u = \sum_{j \in A^{x'}_u} D_{X=x'}(j) =  \sum_{j \in A^{x'}_v} D_{X=x'}(j) = z_v.
\]
On the other hand, in $D_{X=x_m}$, we have
\[
z_u = \sum_{j \in A^{x_m}_u} D_{X=x_m}(j) <  \sum_{j \in A^{x_m}_v} D_{X=x_m}(j) = z_v.
\]
This contradicts the fact that $x_m,x' \in L^i$.
Thus $A^{x'}_u \subset A^{x'}_v$ when both $u,v < i$. 

It is easy to see that for any $x\in \{x' ,x_m\}$, $A^{x}_i \neq A^x_u$ for all $1 \le u < i$ because both $x',x^m \in L^i_3$. Hence, $A^{x_m}_u \subset A^{x_m}_v$ implies $A^{x'}_u \subset A^{x'}_v$ for the case when one of $u,v$ is equal to $i$.
\end{proof}

Let $1 \le u,v\le i$. We say $A^x_u$ is a child of $A^x_v$ (alternately, $A^x_v$ is a parent of $A^x_u$) if $A^x_u \subset A^x_v$, and there is no $w \in [i]$ such that $A^x_u \subset A^x_w \subset A^x_v$. Let $1 \le p,c_1,\dots,c_r,b_1,\dots,b_k < i$. Let $A^x_p$ be a parent of $A^x_i$, and let $A^x_{c_1},\cdots,A^x_{c_r}$ be the children of $A^x_i$. Further, let $A^x_{b_1},\cdots,A^x_{b_k}$ be the children (apart from $A^x_i$) of $A^x_p$. (In other words, $A^x_{b_1},\cdots,A^x_{b_k}$ are siblings of $A^x_i$.) Next, define 
\begin{align*}
    & B^x := A^x_p \setminus \Big(\cup_{j \in [r]}A^x_{c_j} \bigcup \cup_{j\in [k]} A^x_{b_j}\Big),\text{ and}\\
    & C^x :=  A^x_i\setminus (\cup_{j \in [r]}A^x_{c_j}).
\end{align*}

Now by~\autoref{clm:subsetrelation}, for any $x,x' \in L^i_3$, $A^x_u$ is a child (or parent) of $A^x_v$ if and only if $A^{x'}_u$ is a child (or parent) of $A^{x'}_v$.  
Therefore, for all $x \in L^i_3$,
\begin{itemize}
    \item $C^x, B^x \neq \emptyset$, and
    \item $D_X(B^x) = z_p-z_{c_1}-\cdots-z_{c_r}-z_{b_1}-\cdots-z_{b_k} > 0$.
\end{itemize}
Let $ h := z_p-z_{c_1}-\cdots-z_{c_r}-z_{b_1}-\cdots-z_{b_k}$.
Note, 
\begin{align*}
    Z_i|(X=x,Z^{i-1}=z^{i-1}) &= D_X(A^x_i)\\
    &= D_X(A^x_{c_1}) + \cdots +   D_X(A^x_{c_r})+  D_X(C^x) \\
    & =   z_{c_1}+\cdots+z_{c_r}+  D_X(C^x).
\end{align*}
Then we get
\begin{align}
\label{eq:KL-bound}
   KL(Q_{Z_i|X=x, Z^{i-1}=z^{i-1}} || Q_{Z_i|X=x', Z^{i-1}=z^{i-1}}) &= KL(Q_{D_X(A^x_i)} || Q_{D_X(A^{x'}_i)}) \nonumber \\
    & = KL(Q_{D_X(C^x)} || Q_{D_X(C^{x'})})  \nonumber \\
    & = KL\Big(Q_{ \frac{D_X(C^x)}{h}} || Q_{\frac{D_X(C^{x'})}{h}}\Big)
\end{align}
where the last two equalities follow since $z_{c_1}+\cdots+z_{c_r}$ and $h$ depend only on $z^{i-1}$ (not on $x$).

Let us now recall the independence property of Dirichlet distribution (\autoref{lem:independir}). Suppose $(P_1,\dots,P_K) \sim Dir(\alpha_1,\dots,\alpha_K)$. Given $E \subset F \subseteq [K]$ and $d \in (0,1)$, the random variable $\frac{\sum_{i \in E} P_i}{d}$ conditioned on  $\sum_{i \in F}P_i = d$ is distributed a per the  beta distribution with parameter $(\sum_{i \in E} \alpha_i,\sum_{i \in F\setminus E}\alpha_i)$.

For our purpose, we set $d = h$, $F = B^x$ and $E = C^x$ in~\autoref{lem:independir}. Therefore,
\begin{align*}
    & \frac{ D_X(C^x)}{h} \sim Beta(\sum_{i \in C^x} \alpha_i,\sum_{i \in B^x\setminus C^x}\alpha_i),\text{ and}\\
    & \frac{ D_X(C^{x'})}{h} \sim Beta(\sum_{i \in C^{x'}} \alpha_i,\sum_{i \in B^{x'}\setminus C^{x'}}\alpha_i).
\end{align*}

  
  Let for any $x \in L^i_3$,  $U^x = B^x \setminus C^x$, $\alpha(U^x) = \sum_{i \in U^x}\alpha_i$ and $\alpha(C^x) = \sum_{i \in C^x}\alpha_i$. Therefore,
  \begin{align*}
    & \frac{ D_X(C^x)}{h} \sim Beta( \alpha(C^x),\alpha(U^x)),\text{ and}\\
    & \frac{ D_X(C^{x'})}{h} \sim Beta(\alpha(C^{x'}),\alpha(U^{x'})).
\end{align*}
  
  The KL divergence between two beta distributions with different parameters is well known (e.g., see~\cite{johnson1995continuous}) and as follows.

\begin{lemma}
\label{pair-bound}
 \begin{align*}
KL(Beta(\alpha(C^x),\alpha(U^x))|| Beta(\alpha(C^{x'}),\alpha(U^{x'})) & = \ln (\frac{B(\alpha(C^{x'}),\alpha(U^{x'}))}{B(\alpha(C^x),\alpha(U^x))})\\
& +(\alpha(C^x)-\alpha(C^{x'}))\psi(\alpha(C^x))+(\alpha(U^x)-\alpha(U^{x'}))\psi(\alpha(U^x))\\
& + (\alpha(C^{x'})-\alpha(C^{x})+\alpha(U^{x'})-\alpha(D^{x}))\psi(\alpha(C^x)+\alpha(U^x))
\end{align*}
 where $B(w,y) = \frac{\Gamma(w)\Gamma(y)}{\Gamma(w+y)}$ and $\Gamma()$ is a gamma function, and $\psi$ is the digamma function defined as the logarithmic derivative of the gamma function: $\psi(\alpha) = \frac{d}{d \alpha} \ln\Gamma(\alpha)$.
\end{lemma}

 Next, we derive a simple claim about the digamma function.
 \begin{lemma}
 \label{lem:digamma}
 For $ 0 < w \le 1$, we have $|w \psi(w)| \le 3$.
 \end{lemma}
 \begin{proof}
It is well known and follows from the definition of digamma function (e.g., see~\cite{bernardo1976psi}) that
\[
\psi(w) = \psi(w+1) + \frac{1}{w}.
\]
So $|w \psi(w)| \le 1 + w |\psi(w+1)|$. It is also known that $|\psi(w+1)| \le \frac{1}{1+w}+ \ln (1+w)$. So, we get that
\begin{align*}
    |w \psi(w)| &\le 1 + w |\psi(w+1)| \\
    &\le 1 + \frac{w}{1+w}+w \ln (1+w) < 3.
\end{align*}
\end{proof}

\begin{lemma}
\label{lem:boundedkl}
For all $x,x' \in L^i_3$, $KL(Q_{Z_i|X=x,Z^{i-1}= z^{i-1}} || Q_{Z_i|X=x', Z^{i-1}=z^{i-1}}) = O(1)$ .
\end{lemma}
\begin{proof}
W.l.o.g. assume, $x < x' \in L^i_3$. We have $C^x \subset C^{x'}$ and $U^x \subset U^{x'}$. So we have 
\begin{align*}
    |\alpha(C^x)-\alpha(C^{x'})| & \le \sum_{j \ge 2^x+1}\alpha_j \\
    &\le \alpha_{2^x} \le \alpha(C^x)
\end{align*}
where the second last inequality follows since $\alpha_j = \frac{1}{2^j}$ for all $j$. 
Similarly,
\[
|\alpha(U^{x'})-\alpha(D^{x})| \le \alpha(U^x)
\]
From~\autoref{lem:digamma}, we have  
\begin{enumerate}
    \item $|(\alpha(C^x)-\alpha(C^{x'}))| \cdot |\psi(\alpha(C^x))| \le \alpha(C^x) |\psi(\alpha(C^x))| < 3$,
    \item $|(\alpha(U^x)-\alpha(U^{x'}))| \cdot |\psi(\alpha(U^x))| \le \alpha(U^x) |\psi(\alpha(U^x))| < 3$, and
    \item $|(\alpha(C^{x'})-\alpha(C^{x})+\alpha(U^{x'})-\alpha(D^{x}))| \cdot |\psi(\alpha(C^x)+\alpha(U^x))| \le (\alpha(C^x)+\alpha(U^x)) |\psi(\alpha(C^x)+\alpha(U^x))| < 3$.
\end{enumerate}

   So it follows from~\autoref{eq:KL-bound} and~\autoref{pair-bound}, we have 
\[
KL(Q_{Z_i|X=x, Z^{i-1}=z^{i-1}} || Q_{Z_i|X=x', Z^{i-1}=z^{i-1}}) = O(1) + \ln (\frac{B(\alpha(C^{x'}),\alpha(U^{x'}))}{B(\alpha(C^x),\alpha(U^x))}).
\]


It is well known and follows from the definition (e.g., see~\cite{davis1959leonhard}) that
\begin{align*}
    & B(w,y) = \frac{\Gamma(w) \Gamma(y)}{\Gamma(w+y)}, \text{ and}\\
    &\Gamma(w) = \frac{\Gamma(w+1)}{w}.
\end{align*}

Therefore, 
\begin{align*}
    \ln (\frac{B(\alpha(C^{x'}),\alpha(U^{x'}))}{B(\alpha(C^x),\alpha(U^x))}) &= \ln \frac{\alpha(C^x) \alpha(U^x) (\alpha(C^{x'})+\alpha(U^{x'}))}{\alpha(C^{x'}) \alpha(U^{x'}) (\alpha(C^x)+\alpha(U^x))}\\
    & + \ln \frac{\Gamma(1+\alpha(C^{x'})) \Gamma(1+\alpha(U^{x'})) \Gamma(1+\alpha(C^x)+\alpha(U^x))}{\Gamma(1+\alpha(C^x)) \Gamma(1+ \alpha(U^x)) \Gamma(1+\alpha(C^{x'})+\alpha(U^{x'}))}.
\end{align*}

Note, $\frac{\alpha(C^{x'})}{\alpha(C^{x})}, \frac{\alpha(U^{x'})}{\alpha(D^{x})}, \frac{\alpha(C^{x'})+\alpha(U^{x'})}{\alpha(C^{x}+\alpha(D^{x})}\in [1,2]$. So we have 
\[
\Big|\ln \frac{\alpha(C^x) \alpha(U^x) (\alpha(C^{x'})+\alpha(U^{x'}))}{\alpha(C^{x'}) \alpha(U^{x'}) (\alpha(C^x)+\alpha(U^x))} \Big| \le 3 \ln 2.
\]
For any $0 \le w \le 1$, it is known that $\Gamma(1+w) \in [0.6,2]$. Therefore, 
\[
\Big|\ln \frac{\Gamma(1+\alpha(C^{x'})) \Gamma(1+\alpha(U^{x'})) \Gamma(1+\alpha(C^x)+\alpha(U^x))}{\Gamma(1+\alpha(C^x)) \Gamma(1+ \alpha(U^x)) \Gamma(1+\alpha(C^{x'})+\alpha(U^{x'}))} \Big| = O(1).
\]
Thus we have $ \ln (\frac{B(\alpha(C^{x'}),\alpha(U^{x'}))}{B(\alpha(C^x),\alpha(U^x))}) = O(1)$, and hence $KL(Q_{Z_i|X=x, Z^{i-1}=z^{i-1}} || Q_{Z_i|X=x', Z^{i-1}=z^{i-1}}) = O(1)$.
\end{proof}

Now, it directly follows from~\autoref{upperboundonI} and~\autoref{lem:boundedkl} that for every $i \in [t]$, $z^{i-1}$, and $x \in [\log n]$, $I(X;Z_i|Z^{i-1}=z^{i-1})  = O(1)$ which along with~\autoref{clm:Igivenz} completes the proof of~\autoref{lem:boundonIstep}.

\subsection{Extending $\Omega(\log \log \log n)$ bound to {\CD} model}
\label{sec:cd-LB}
We now extend the $\Omega(\log \log \log n)$ bound to the ${\CD}$ model. 
The extra difficulty (compared to the {\SEval}) comes from the fact that now at any step, there are three outcomes - (i) the sum $Z_i = D_X(A_i)$, (ii) the sampled element $E_i \in A_i$ (if $D_X(A_i) > 0$), and (iii) the actual probability $D_X(E_i)$ of $E_i$ (in actual definition, the outcomes are $(E_i, D_X(E_i),\frac{D_X(E_i)}{D_X(A_i)})$. Obviously it is equivalent to $(E_i, D_X(E_i), D_X(A_i))$). If $D_X(A_i) = 0$ then we set $E_i$ to be $\emptyset$. We need to show now that the information gain about $X$ from all of these outcomes is small. Formally, we need to upper bound $I(X;(Z_i,E_i,D_X(E_i)))$ (compared to earlier $I(X;Z_i)$). Further, (as before), the upper bound $I(X;(Z_i,E_i,D_X(E_i)))$ needs to be shown conditioned on the previous outcomes. 

The set of inputs $\mathcal{X}$ and distribution $\gamma$ is the same as before except that we set the Dirichlet parameters $\alpha_1,\dots,\alpha_n$ as follows:
\[
\alpha_{j} = \frac{1}{c^{j}} \text{ for all }1 \le j \le n \text{, where }c = 4 (2^{2^{n}})^2 2^{2^{n}}.
\]

As before, we will require that the queried sets are laminar.
\begin{lemma}
\label{lem:laminarappendix}
Any algorithm that makes $t$ queries of {\CD} can be simulated by an algorithm that makes at most $2^t$ {\CD}  queries with the property that   (i)the queried sets are laminar, and (ii) a set queried later is not a strict superset of a set queried as well as sampled elements earlier (i.e., of those sets $S \subseteq [n]$, whose $D_X(S)$ is known so far). 
\end{lemma}
\begin{proof}
Let algorithm $\mathcal{A}$ makes $t$ queries of {\CD}. Let the queried sets by $\mathcal{A}$, in order, be denoted by $A_1,\dots, A_t$.  We give an algorithm $\mathcal{B}$ that can simulate $A$ such that the queried sets satisfy the property as mentioned in the lemma. Like before, we use $D_X(A)$ for $\sum_{j \in A}D_X(j)$ for all $A \subseteq [n]$. The first set queried by $\mathcal{B}$ is also $A_1$.  The algorithm $\mathcal{B}$ receives the tuple $(D_X(A_1),E(A_1),D_X(E(A_1))$ (where $E(A_1) \in A_1$ is the sampled element) and via $\mathcal{A}$ determines  the next set $A_2$ queried by $\mathcal{A}$. 

Let $\mathcal{E}_{i-1}$ be the set of all sampled elements of $\mathcal{B}$ till step $i-1$, i.e., till $\mathcal{B}$ has determined the set $A_i$ by simulating the sets $A_1,\dots,A_{i-1}$.  Let $\mathcal{E}_{A_i} = A_i \cap \mathcal{E}_{i-1}$ be the set of sampled elements so far in $A_i$. In place of set $A_i$, the algorithm $\mathcal{B}$ queries all sets in
\[
\mathcal{C}_i =  \{A_i \cap C_1 \cap C_2 \dots C_{i-1} \cap (\cap_{E \in \mathcal{E}_{A_i}} E^c)|C_j \in \{A_j,A^c_j\}, \forall j \in [i-1]\}
\]
($A^c_j = [n]\setminus A_j$ and $E^c = [n]\setminus E$).  For each set $C \in \mathcal{C}_i$, the algorithm $\mathcal{B}$ receives the tuple $(D_X(C),E(C),D_X(E(C)))$, where $E(C) \in C$ is the sampled element with probability $\frac{D_X(E(C))}{D_X(C)}$. Note that $\mathcal{B}$ can determine $D_X(A_i)$ since
\[
D_X(A_i) =\sum_{C \in \mathcal{C}_i} D_X(C) + \sum_{E \in \mathcal{E}_{A_i}} D_X(E).
\]
Now $\mathcal{B}$ needs to sample an element $E_i \in A_i$ with probability $\frac{D_X(E_i)}{D_X(A_i)}$. For the same, the algorithm $\mathcal{B}$ samples a set $U \in \mathcal{C}_{i} \cup \mathcal{E}_{A_i}$ with probability $\frac{D_X(U)}{D_X(A_i)}$ (note that $D_X(U)$ is known to $\mathcal{B}$ for all $U \in \mathcal{C}_{i} \cup \mathcal{E}_{A_i}$) and then to determine the next set $A_{i+1}$, uses the tuple $(D_X(A_i),E(U),D_X(E(U)))$ or $(D_X(A_i),U,D_X(U))$ as an outcome of $\mathcal{A}$  for $A_i$, depending on $U \in \mathcal{C}_i $ or $U \in \mathcal{E}_{A_i}$ respectively. 

Observe, $|\mathcal{C}_i| = 2^{i-1}$, and hence the total number of queries is bounded by $\sum_{i \in [t]}2^{i-1} \le 2^t.$ Also, it is easy to see that the queried sets satisfy the conditions of the lemma.


\end{proof}

Let $m_i \in A_i$ be the minimum  element of $A_i$. We will need the following lemma later to show that for any $A_i$ ($1 \le i \le t$), if $D_X(A_i) > 0$, the sampled element $E_i$ will be $m_i$ with probability at least $1 - \frac{1}{2^{2^n}}$.

\begin{lemma}
\label{lem:beta}
Consider a random variable $Y \sim Beta(\alpha,\beta)$ where $\beta \le \frac{\alpha}{c}$ and $c \ge  4(2^{2^{n}})^2 2^{2^{n}}$. We have $\Pr[ Y \le 1 - \frac{1}{2^{2^{n}}}] \le \frac{1}{2^{2^{n}}}.$ 
\end{lemma}
\begin{proof}
Let $d =  4(2^{2^{n}})^2 2^{2^{n}}$. It is immediate from the property of beta distributions, that
\begin{align*}
    & \E[Y] = \frac{\alpha}{\alpha+ \beta} \ge 1 - \frac{1}{d+1} \text{ and}\\
    & Var[Y] =  \frac{\alpha \beta}{(\alpha + \beta)^2 (\alpha+\beta+1)} \le \frac{1}{d}.
\end{align*}
Therefore, we derive that
\begin{align*}
    \Pr[Y \le 1 - \frac{1}{2^{2^{n}}}] & \le \Pr[|Y - \E[Y]| > \frac{1}{2 \cdot 2^{2^n}}] \\
    & < 4 (2^{2^n})^2 Var[Y] &&\text{(by Chebyshev inequality)}\\
    &\le \frac{4 (2^{2^n})^2 }{4(2^{2^{n}})^2 2^{2^{n}}} = \frac{1}{2^{2^{n}}}.
\end{align*}
\end{proof}

 We will use $H^{i-1}$ to denote  the history $(Z_{i-1},E_{i-1},D_X(E_{i-1}),\dots,Z_1,E_1,D_X(E_1)).$
  Then observe, 
  \begin{align}
  \label{eq:mutual-triplet}
I(X;(Z_i,E_i,D_X(E_i))|H^{i-1})& = I(X;Z_i|H^{i-1}) + I(X;E_i|Z_i,H^{i-1})+I(X;D_X(E_i))|E_i,Z_i,H^{i-1}).
  \end{align}
  \begin{lemma}
  \label{lem:first}
  $I(X;Z_i|H^{i-1}) = O(1)$.
  \end{lemma}
  \begin{proof}
The distribution of $Z_i$ will depend on $Z_1,\dots,Z_{i-1}$ and on $D_X(E_1),\dots,D_X(E_{i-1})$, but not on the fact that element $E_j$ was sampled from set $A_j$ for $j \le i-1$. So the proof of the current lemma is the same as that of~\autoref{lem:boundonIstep}, except that now the Dirichlet parameters are set to $\alpha_{j} = \frac{1}{c^{j}}$ for all $1 \le j \le n$. However, all the previous calculations (as in the proof of~\autoref{lem:boundonIstep}) will also work here.

  \end{proof}
  \begin{lemma}
  \label{lem:second}
  $I(X;E_i|Z_i,H^{i-1}) = O(1)$.
  \end{lemma}
  \begin{proof}
To show  $I(X;E_i|Z_i,H^{i-1}) = O(1)$, as previously, it suffices to show, for given any $z_i,h^{i-1}$,
\[
I(X;E_i|Z_i=z_i,H^{i-1}=h^{i-1}) = O(1)
\]
where 
\[
I(X;E_i|Z_i=z_i,H^{i-1}=h^{i-1})= \sum_{x \in L^c} \sum_{e_i}Q_{X,E_i|Z_i,H^{i-1}}(x,e_i|z_i,h^{i-1}) \log \frac{Q_{X,E_i|Z_i,H^{i-1}}(x,e_i|z_i,h^{i-1})}{Q_{X|Z_i,H^{i-1}}(x|z_i,h^{i-1}) Q_{E_i|Z_i,H^{i-1}}(e_i|z_i,h^{i-1})}
\]
  From now, we will omit the $z_i,h^{i-1}$ in the density functions (as it is understood). Further, for brevity, we will omit the subscript when the context is understood. For example, we will use $q(x,e_i) $ for $Q_{X,E_i|Z_i,H^{i-1}}(x,e_i|z_i,h^{i-1})$.
 Now for any $x \in L^i_0,$ we have $e_i = \emptyset$ with probability one. Therefore, 
   \[ \sum_{x  \in L^i_0} \sum_{e_i}q(x,e_i) \log \frac{q(x,e_i)}{q(s) q(e_i)}  = \Big(\sum_{x  \in L^i_0} g_x \Big) \log \frac{1}{\sum_{x \in L^i_0}g_x}.
 \]

Since the set $A_i$ does not contain any of the earlier sets $A_1,\dots,A_{i-1}$ as well as the sampled elements $E_1,\dots,E_{i-1}$, we have
\[
\frac{D_X(m_i)}{z_i} \sim Beta(\alpha(m_i),\alpha(A_i\setminus m_i)).
\]
Therefore, by~\autoref{lem:beta}, 
   for $x \not \in L^i_0$, we have $\Pr[E_i \neq m_i] \le \frac{1}{2^{2^n}}$. Therefore,
 \[
 \sum_{x \not \in L^i_0} \sum_{e_i}q(x,e_i) \log \frac{q(x,e_i)}{q(x) q(e_i)} \le \sum_{x \not \in L^i_0} q(x,m_i) \log \frac{q(x,m_i)}{q(x) q(m_i)} +  (n \log n)   |q(x,e) \log \frac{q(x,e)}{q(x) q(e)}|
\]
 where $e$ is an any element in $A_i\setminus m_i$. Now,
\begin{align*}
    |q(x,e) \log \frac{q(x,e)}{q(x) q(e)}| & \le |q(x,e)\log q(x,e)| + |q(x,e) \log q(x)| + |q(x,e) \log q(e)| \\
    &\le 3 |q(x,e)\log q(x,e)|.
\end{align*}
 
 As $q(x,e) = g_x q_x(e) \le \frac{1}{2^{2^n}}$, we have $|q(x,e)\log q(x,e)| \le \frac{2^n}{2^{2^n}}.$ Therefore,
 \begin{align*}
 \sum_{x \not \in L^i_0} \sum_{e_i}q(x,e_i) \log \frac{q(x,e_i)}{q(x) q(e_i)} &\le \sum_{x \not \in L^i_0} q(x,m_i) \log \frac{q(x,m_i)}{q(s) q(m_i)} +  (3n \log n)\frac{2^n}{2^{2^n}}\\
 & =  \sum_{x \not \in L^i_0} g_x q_x(m_i) \log \frac{q_x(m_i)}{ \sum_{x \not \in L^i_0}g_x q_x(m_i)} +(3n \log n)\frac{2^n}{2^{2^n}} \\
&\le  \sum_{x \not \in L^i_0} g_x \log \frac{1}{ \sum_{x \not \in L^i_0}g_x (1 - \frac{1}{2^{2^n}})} + \sum_{x \not \in L^i_0} g_x |\log \frac{1-\frac{1}{2^{2^n}}}{ \sum_{x \not \in L^i_0}g_x }|+ (3n \log n)\frac{2^n}{2^{2^n}}\\
& = \sum_{x \not \in L^i_0} g_x \log \frac{1}{\sum_{x \not \in L^i_0} g_x} +o(1).\\
 \end{align*}
 Also, 
 \[ \sum_{x  \in L^i_0} \sum_{e_i}q(x,e_i) \log \frac{q(x,e_i)}{q(s) q(e_i)}  = \sum_{x  \in L^i_0} g_x \log \frac{1}{\sum_{x \in L^i_0}g_x}.
 \]
 
 Hence, we have
 \begin{align*}
     I(X;E_i|z_i,h^{i-1}) &= \sum_{x \not \in L^i_0} \sum_{e_i}q(x,e_i) \log \frac{q(x,e_i)}{q(x) q(e_i)}+\sum_{x  \in L^i_0} \sum_{e_i}q(x,e_i) \log \frac{q(x,e_i)}{q(s) q(e_i)} \\
     &\le \log 2 + o(1) = O(1)
 \end{align*}
 which completes the proof.
 \end{proof}
 
 \begin{lemma}
 If for any value of $z_i$ and $h^{i-1}$, we have $I(X;D_X(E_i)|E_i = m_i,Z_i = z_i,H^{i-1} = h^{i-1}) = O(1)$ then $I(X;D_X(E_i)|E_i,Z_i,H^{i-1}) = O(1)$.
 \end{lemma}
  \begin{proof}
  We have 
  \begin{align*}
      I(X;D_X(E_i)|E_i,Z_i,H^{i-1}) &= \sum_{e_i} \int_{(z_i,h^{i-1})} q(e_i,z_i,h^{i-1}) I(X;D_X(E_i)|E_i=e_i,Z_i=z_i,H^{i-1}=h^{i-1}) \,d(z_i,h^{i-1})\\
      &  = \int_{(z_i,h^{i-1})} q(\emptyset,z_i,h^{i-1}) I(X;D_X(E_i)|E_i=\emptyset,Z_i=z_i,H^{i-1}=h^{i-1}) \,d(z_i,h^{i-1}) \\
      &+ \int_{(z_i,h^{i-1})} q(m_i,z_i,h^{i-1}) I(X;P_i|E_i=0,Z_i=z_i,H^{i-1}=h^{i-1}) \,d(z_i,h^{i-1})\\
      &+ \sum_{e_i \not \in \{0,m_i\}} \int_{(z_i,h^{i-1})} q(e_i,z_i,h^{i-1}) I(X;D_X(E_i)|E_i=e_i,Z_i=z_i,H^{i-1}=h^{i-1}) \,d(z_i,h^{i-1}).
  \end{align*}
  
  Note that
  \begin{align*}
      I(X;D_X(E_i)\mid E_i = \emptyset,Z_i = z_i,H^{i-1} = h^{i-1}) &= \Big(\sum_{x \in L^i_0} g_x\Big) \log \frac{1}{\sum_{x \in L^i_0} g_x} \\
      &= O(1).
  \end{align*}
  Then, by the assumption of the lemma, $I(X;P_i|E_i = m_i,Z_i = z_i,H^{i-1} = h^{i-1}) = O(1)$. Therefore,  
  \begin{align*}
      & I(X;P_i|E_i,Z_i,H^{i-1}) \\
     & \le  O(1) +  \sum_{e_i \not \in \{0,m_i\}} \int_{(z_i,h^{i-1})} q(e_i,z_i,h^{i-1}) I(X;D_X(E_i)|E_i=e_i,Z_i=z_i,H^{i-1}=h^{i-1}) \,d(z_i,h^{i-1}) \\
     &\le O(1)+ \sum_{e_i \not \in \{0,m_i\}} \int_{(z_i,h^{i-1})}  q(e_i|z_i,h^{i-1}) q(z_i,h^{i-1}) I(X;D_X(E_i)|E_i=e_i,Z_i=z_i,H^{i-1}=h^{i-1}) \,d(z_i,h^{i-1}) \\
     & \le O(1)+ \sum_{e_i \not \in \{0,m_i\}} \int_{(z_i,h^{i-1})}  \frac{1}{2^{2^n}} q(z_i,h^{i-1}) I(X;D_X(E_i)|E_i=e_i,Z_i=z_i,H^{i-1}=h^{i-1}) \,d(z_i,h^{i-1}) \\
     &\le O(1) +  \frac{n}{2^{2^n}} \int_{(z_i,h^{i-1})}   q(z_i,h^{i-1}) I(X;D_X(E_i)|E_i=e_i,Z_i=z_i,H^{i-1}=h^{i-1}) \,d(z_i,h^{i-1}). 
  \end{align*}
  Note,
\begin{align*}
    & I(X;D_X(E_i)|E_i=e_i,Z_i=z_i,H^{i-1}=h^{i-1}) \\
    & = \sum_{x} \int_{p_i} q(x,p(e_i)|e_i,z_i,h^{i-1}) \log \frac{q(x,p(e_i)|e_i,z_i,h^{i-1})}{q(x|e_i,z_i,h^{i-1}) q(p(e_i)|e_i,z_i,h^{i-1})} \,d(q_i)\\
    &\le  3 \log n .
\end{align*}

  Hence,  $I(X;D_X(E_i)|E_i,Z_i,H^{i-1}) \le O(1) + \frac{3n \log n}{2^{2^n}} = O(1)$.
  \end{proof}

It remains to argue that $I(S;D_X(m_i)|E_i = m_i,Z_i,H^{i-1}) = O(1)$. Unfortunately, $I(S;D_X(m_i)|E_i = m_i,Z_i,H^{i-1}) \neq I(S;D_X(m_i)|Z_i,H^{i-1})$, and thus~\autoref{lem:boundonIstep} cannot be used directly. However, the following well-known property of Dirichlet distribution will help us to remove the condition $E_i = m_i$.

  
  
 
\begin{proposition}
\label{lem:postsample}
Suppose $(P_1,\dots,P_K) \sim Dir(\alpha_1,\dots,\alpha_K)$. Let $e \in [K]$. Let $E$ be the event that $e$ is  the element sampled from the distribution $(P_1,\dots, P_K)$. Let  for $1 \le i \le K$, $G_i$ takes the value $1$ if $i = e$, and $0$ otherwise. Then  $(P_1,\dots,P_K|E) \sim Dir(\alpha_1+G_1,\dots,\alpha_K+G_K)$.
\end{proposition}  

 Now, we are ready to conclude the following.
\begin{lemma}
\label{lem:third}
$I(X;D_X(m_i)|E_i = m_i,Z_i = z_i,H^{i-1} = h^{i-1}) = O(1).$
\end{lemma}
\begin{proof}
It suffices to show that for any $x < x' \in \Bar{L}$, we have
\begin{align*}
    KL(q_x(p(m_i)|m_i,z_i,h^{i-1})||q_{x'}(p(m_i)|m_i,z_i,h^{i-1})) &= KL(q_x(\frac{p(m_i)}{z_i}|m_i,z_i,h^{i-1})||q_{x'}(\frac{p(m_i)}{z_i}|m_i,z_i,h^{i-1})) \\
    &= O(1).
\end{align*}

In our case, $C^x = C^{x'} = m_i$, and $D^x$, $D^{x'}$ are subsets of $A^x_i \setminus m_i$ and $A^{x'}_i \setminus m_i$ respectively. 
Recall that
\[
\Big(\frac{D_X(m_i)}{z_i}|x,Z_i = z_i,H^{i-1} = h^{i-1}\Big) \sim Beta(\alpha(m_i),\alpha(D^x))
\]
for any $x \in \Bar{L}$. From~\autoref{lem:postsample}, we have
\[
\Big(\frac{D_X(m_i)}{z_i}|x,E_i = m_i,Z_i = z_i,H^{i-1} = h^{i-1}\Big) \sim Beta(1+\alpha(m_i),\alpha(D^x)).
\]
All we now need to show that the calculations in~\autoref{lem:boundedkl} will work for the current values of Dirichlet parameters --- $\alpha_{m_i} = 1+ \frac{1}{c^{m_i}}$ and $\alpha_j = \frac{1}{c^j}$ for $j \ge m_i +1$ (note, $\alpha_1,\dots,\alpha_{m_i -1}$ do not appear in  the calculations).

The inequalities used in~\autoref{lem:boundedkl} are
\begin{enumerate}
    \item[(i)] $|\alpha(C^x)-\alpha(C^{x'})| \le \sum_{j \ge 2^x+1}\alpha_j  \le \alpha_{2^x} \le \alpha(C^x)$ (this now holds trivially since $\alpha(C^x) = \alpha(C^{x'}$),
    \item[(ii)] $|\alpha(D^{x'})-\alpha(D^{x})| \le \sum_{j \ge 2^x+1}\alpha_j  \le \alpha_{2^x} \le \alpha(D^x)$ (since $m_i$ is neither in $D^x$ nor in $D^{x'}$, this also holds now),
    \item[(iii)] $\frac{\alpha(C^{x'})}{\alpha(C^{x})}, 
\frac{\alpha(D^{x'})}{\alpha(D^{x})}, \frac{\alpha(C^{x'})+\alpha(D^{x'})}{\alpha(C^{x}+\alpha(D^{x})}\in [1,2]$ (easy to see all of them holds now), and
    \item[(iv)] $|\ln \frac{\Gamma(1+\alpha(C^{x'})) \Gamma(1+\alpha(D^{x'})) \Gamma(1+\alpha(C^x)+\alpha(D^x))}{\Gamma(1+\alpha(C^x)) \Gamma(1+ \alpha(D^x)) \Gamma(1+\alpha(C^{x'})+\alpha(D^{x'}))} | = O(1)$ (which also holds now as for $0 \le w \le 2$,  $\Gamma(1+w) \in [0.6,2]$).
\end{enumerate}
The lemma now follows by following the calculations in~\autoref{lem:boundedkl}.
\end{proof}
 From~\autoref{eq:mutual-triplet}, and~\autoref{lem:first},~\ref{lem:second} and ~\ref{lem:third}, we have $I(X;(Z_i,E_i,D_X(E_i))|H^{i-1}) = O(1)$. Hence, as in Section~\ref{sec:L1lowerbound}, by chain rule, we get
 \[
 I(X;(Z_1,E_1,D_X(E_1),\dots,(Z_t,E_t,D_X(E_t))) \le  O(t).
\]
Hence, by Fano's inequality (\autoref{thm:fano}), we have the error probability at least $1 - \frac{O(t)}{\log \log n}$. Hence, $t = \Omega(\log \log n)$. Recall that this lower bound is for the tester that queries only the laminar family of sets. Finally, as in Section~\ref{sec:L1lowerbound}, using~\autoref{lem:laminarappendix}, we get a lower bound of $ \Omega(\log \log \log n)$ for general testers. So we conclude the following.

\cdlowerbound*

\paragraph{Discussion on stronger lower bound with approximate {\CD} queries. }Now, we discuss how we get an $\Omega(\log \log n)$ query lower bound when we have access to a {\CE} oracle and an approximate {\SEval} oracle. We refer to the oracle that, given a subset $S\subseteq [n]$, returns the value of $D(S)$ up to a 2-multiplicative factor as the approximate {\SEval} oracle. Here, we would like to highlight how, by slightly modifying the proof of~\autoref{lemma:loglogn}, we get an $\Omega(\log \log n)$ lower bound for this extended model. The proof of~\autoref{lemma:loglogn} will remain the same, except now we need to modify the encoding argument slightly. More specifically, for any $A_i \subseteq [n]$, let $m_i \in A_i$ be the minimum element (minimum is well-defined since we consider the ordered domain is $[n]$) in $A_i$. Note that for any $x \in [\log n]$, we have $D_x(m_i) \ge  D_x(A_i)/2$ (for the choice of hard distribution $D_x$ in the proof of~\autoref{lemma:loglogn}). Thus, Bob can determine the approximate value of $D_x(A_i)$ without any extra bit sent by Alice. In general, for $(1+\delta)$-multiplicative approximation of the {\SEval} query, Alice needs to send only $O(1/\delta)$ extra bits to Bob to specify the estimate. Hence, for any constant $\delta$, we get the same lower bound of $\Omega(\log \log n)$ as in~\autoref{lemma:loglogn}. This lower bound result is interesting on its own. Apart from that, it also provides a piece of evidence that $\Omega(\log \log n)$ might be the correct lower bound even for (exact) {\CD} (which is a combination of {\CE} and exact {\SEval}) queries. Therefore, we believe that our Conjecture~\ref{conj:cdlowerbound} is true, i.e., $\Omega(\log \log n)$ lower bound holds even for the {\CD} model.

\subsection{An $\Omega(\frac{1}{\epsilon^2 \log (1/\epsilon)})$ lower bound}
\label{sec:additiveLB}
Now we prove a lower bound of $\Omega(\frac{1}{\epsilon^2 \log (1/\epsilon)})$ for additive $\epsilon n$-approximation of support size (\autoref{cor:lbepsilon}). 

\begin{lemma}
\label{lem:lbepsilon}
Any algorithm that, given {\CD} access to a distribution $D$ on $[n]$ and any $\epsilon=\epsilon(n) \in (0,1]$ where $\epsilon(n) \ge \Omega(n^{-1/2})$, estimates the support size $|\supp(D)|$ to an additive $\epsilon n$-factor with probability at least $\frac{2}{3}$, must make $\Omega(\frac{1}{\epsilon^2 \log n})$ queries to the {\CD} oracle.
\end{lemma}

The proof is based on a reduction from a well-studied communication complexity problem, namely the \emph{Gap-Hamming distance} problem, to the support size estimation problem.

For any two strings $x,y \in \{0,1\}^n$, their Hamming distance, denoted by $d_H(x,y)$, is the number of coordinates where the bit-values of $x,y$ differ, i.e., $d_H(x,y) := |\{i \in [n] \mid x[i] \ne y[i]\}|$. In the \emph{Gap-Hamming distance} problem $\GHD_{n,n/2,g}$,  Alice and Bob are given strings $x,y \in \{0,1\}^n$ respectively (neither of them knows the input of the other party). Their task is to decide between the following two cases:
\begin{itemize}
    \item Yes: $d_H(x,y) \ge n/2 + g$;
    \item No: $d_H(x,y) < n/2 - g$.
\end{itemize}

In this subsection, we consider the randomized two-way communication complexity of the above problem. In a randomized two-way communication protocol, both Alice and Bob are given access to a shared random string (public randomness) and are allowed to communicate with each other (unlike one-way protocol as considered in Section~\ref{sec:lowerbound}, where only Alice was allowed to send a message). Then the question is how many bits need to be sent to solve the above problem with a probability at least $2/3$. The two-way randomized communication complexity is defined to be the minimum number of bits that is communicated in any protocol that solves $\GHD_{n,n/2,g}$.

\begin{theorem}[\cite{chakrabarti2012optimal}]
\label{thm:GHDcommunication}
The two-way randomized communication complexity of $\GHD_{n,n/2,\sqrt{2n}}$ is $\Omega(n)$.
\end{theorem}

By using a standard padding argument as in~\cite{brody2010better, blais2012property}, the above theorem can be generalized as follows.
\begin{theorem}[\cite{chakrabarti2012optimal, brody2010better, blais2012property}]
\label{thm:GHDcommunication-gen}
For any $g \ge \Omega(\sqrt{n})$, the two-way randomized communication complexity of $\GHD_{n,n/2,g}$ is $\Omega((n/g)^2)$.
\end{theorem}

Now we use the above theorem to prove~\autoref{lem:lbepsilon}.

\begin{proof}[Proof of~\autoref{lem:lbepsilon}]
Suppose there is a randomized algorithm $T$ that, given any $\epsilon > 0$ and a distribution on $[n]$, approximates the support size to an additive $\epsilon n$-factor with probability at least $2/3$ by making at most $t=t(n)$ {\CD} queries. Now consider an instance of $\GHD_{n,n/2,g}$, for any $g \ge \Omega(\sqrt{n})$. Alice and Bob are given $x,y \in \{0,1\}^n$, respectively. From $x$, construct the set $I_x \subseteq [n]$ by including each $i\in [n]$ such that $i$-th bit of $x$ is set to 1. So $I_x = \{i \in [n] \mid x[i]=1\}$. Similarly, construct $I_y \subseteq [n]$ from $y$.

Next, consider the following distribution $D$ over $[n]$
\begin{align*}
   D(j) =
\begin{cases}
\frac{2}{|I_x|+|I_y|} &\text{ if } j\in I_x \cap I_y\\
\frac{1}{(|I_x|+|I_y|)} &\text{ if } j \in I_x \triangle I_y\\
0&\text{ otherwise}
\end{cases} 
\end{align*}
where $\triangle$ denotes the symmetric difference between two sets.

Clearly, $|\supp(D)| = |I_x \cup I_y| = \frac{1}{2} (|I_x|+|I_y|+|I_x \triangle I_y|)$. By the construction of $I_x,I_y$, $|I_x \triangle I_y| = d_H(x,y)$. So
\begin{equation}
    \label{eq:distsupp}
    |\supp(D)| = |I_x \cup I_y| = \frac{1}{2} (|I_x|+|I_y|+d_H(x,y)).
\end{equation}

Next, we argue that Alice and Bob can simulate the algorithm $T$ by setting $\epsilon = \frac{g}{3n}$, using their shared public random string, say $r$, to estimate $|\supp(D)|$, which in turn help them in deciding whether $d_H(x,y) \ge n/2 + g$ or $d_H(x,y) < n/2 - g$. Note, since $I_x, I_y$ are known only by Alice and Bob, respectively, none of them know the distribution $D$.

Let us first argue that given any $S \subseteq [n]$ (known to both Alice and Bob), a {\CD} query to $D$ condition on $S$ can be simulated by Alice and Bob using only $O(\log n)$ bits of communication. Recall performing a {\CD} query means they need to sample from $D$ condition on $S$ and determine the actual and conditional probability of the sampled element. This can be done as follows. Alice first picks an element (say $a$) uniformly at random from $S \cap I_x$ and sends $a$ and the value of $|S \cap I_x|$ to Bob. Bob also selects an element (say $b$) uniformly at random from $S \cap I_y$. Bob then picks $a$ with probability $\frac{|S \cap I_x|}{|S \cap I_x|+|S \cap I_y|}$ and $b$ with probability $\frac{|S \cap I_y|}{|S \cap I_x|+|S \cap I_y|}$. Bob then sends this chosen element $c \in \{a,b\}$ to Alice. It is easy to see that $c$ is picked with a probability equal to $D(c)/D(S)$ (i.e., the conditional probability of picking $c$ from $D$ conditioned on $S$). Further, if $c = a$, then Bob tells Alice whether $a \in I_y$ or not using $O(1)$ bits. If $c = b$, then Alice can by herself check whether $c \in I_x$ or not. In any case, Alice knows if $c \in I_x \cap I_y$ or not. Recall, the actual probability of any element $j$ in the support of $D$ takes only two values -- $\frac{2}{|I_x|+|I_y|}$ if $j \in I_x \cap I_y$, and $\frac{1}{(|I_x|+|I_y|)}$ otherwise. Thus Alice can determine the actual probability of sampled element $c$ as she knows whether $c \in X \cap Y$ or not. Note, 
\[
D(S) = \sum_{j \in S}D(j) = \frac{|S \cap I_x| +| S \cap I_y|}{|I_x|+|I_y|}.
\]
Both Alice and Bob can know the values of $|I_x|,|I_y|,|S \cap I_x|,|S \cap I_y|$ and thus $D(S)$, by communicating only $O(\log n)$ bits. Hence, they can determine the conditional probability of the sampled element as well.

Consider $\epsilon=\frac{g}{3n}$. Let us now describe how Alice and Bob simulate $T$ using their shared random string $r$. Given $r$, the first set $A_1$ on which $T$ conditions (i.e., places {\CD} query), is fixed. Then Alice and Bob simulate this query by communicating $O(\log n)$ bits using the protocol described above. Using the outcome of the first {\CD} query, $T$ decides the second set $A_2$ to condition. In general, at any step $\ell$, $T$ decides the set $A_\ell$ to condition on $D$ based on the outcomes of {\CD} queries so far and the random string $r$. Since $T$ places at most $t$ queries, Alice and Bob need to communicate $O(t \log n)$ bits to simulate $T$. At the end of the simulation, with probability at least $2/3$, Alice and Bob get an estimate $\hat{s}$ of $|\supp(D)|$ up to an $\epsilon n$ additive factor. So, $|\supp(D)| - \epsilon n \le \hat{s} \le |\supp(D)| + \epsilon n$.

For a Yes instance (i.e., $d_H(x,y) \ge n/2 + g$), by~\autoref{eq:distsupp}, $|\supp(D)|\ge \frac{1}{2}(|I_x|+|I_y|+n/2+g)$. Thus the estimated value $\hat{s} \ge \frac{1}{2}(|I_x|+|I_y|+n/2 + g) - \epsilon n$. On the other hand, for a No instance (i.e., $d_H(x,y) < n/2 - g$), by~\autoref{eq:distsupp}, $|\supp(D)| < \frac{1}{2}(|I_x|+|I_y|+n/2-g)$. Thus the estimated value $\hat{s} < \frac{1}{2}(|I_x|+|I_y|+n/2-g) + \epsilon n$. It is now straightforward to verify that for $\epsilon = \frac{g}{3 n}$,
\[
\frac{1}{2}(|I_x|+|I_y|+n/2 - g) + \epsilon n < \frac{1}{2}(|I_x|+|I_y|+n/2 + g) - \epsilon n.
\]
Note both Alice and Bob can know the values of $|I_x|,|I_y|$ by communicating only $O(\log n)$ bits at the beginning. Hence, Alice and Bob can decide whether the given instance of $\GHD_{n,n/2,g}$ is a Yes instance or No instance with probability at least $2/3$ and the total number of bits communicated is at most $O(t \log n)$. Recall, $\epsilon = \frac{g}{3 n}$. Then by~\autoref{thm:GHDcommunication-gen},
\[
t \ge \Omega \Big( \frac{n^2}{g^2 \log n} \Big) = \Omega\Big(\frac{1}{\epsilon^2 \log n} \Big).
\]
\end{proof}

Now, as a direct corollary of~\autoref{lem:lbepsilon}, we get the following.
\begin{corollary}
\label{cor:lbepsilon}
Consider any $\alpha \in (0,1/2) $. Any algorithm that, given {\CD} access to a distribution $D$ on $[n]$ and any $\epsilon=n^{-\alpha}$, estimates the support size $|\supp(D)|$ to an additive $\epsilon n$-factor with probability at least $\frac{2}{3}$, must make $\Omega(\frac{1}{\epsilon^2 \log (1/\epsilon)})$ queries to the {\CD} oracle.
\end{corollary}

\section{Power of  {\CD} model}
\label{sec:power-SEval}
In this section, we formally illustrate the power of {\CD} model. We actually focus on an even weaker model, namely {\SEval}. Recall, a {\SEval} oracle for a distribution $D$ on $[n]$, on input a subset $S \subseteq [n]$, returns the value of $D(S)$. Since using a {\CD} query we can simulate a {\SEval} query (see Section~\ref{sec:prelims}), any upper bound (on the number of queries) in the {\SEval} model also provides the same upper bound on the {\CD} model. 

\subsection*{Equivalence Testing}
As the first example, we consider the standard \emph{equivalence testing} problem. Given access to two (unknown) distributions $D,D'$, this problem asks whether $D=D'$ or they are "far" from each other. More specifically, the goal is to accept if $D = D'$ and reject if their \emph{total variation distance} $||D-D'||_{TV} = \sum_{i \in [n]}|D(i)-D'(i)| > \epsilon$ (both with high probability). It is known that $\Omega(\sqrt{\log \log n})$ queries are necessary in the {\CO} model~\cite{acharya2015chasm}. On the other hand, it is easy to see that $\Omega(1/\epsilon)$ queries must be needed in the {\CE} model. Below we show that for this problem, it suffices to place only two {\SEval} queries to both $D$ and $D'$. In particular, we prove the following.
\begin{theorem}
\label{thm:seteval-equivalence}
There is an algorithm that given {\SEval} access to distributions $D$ and $D'$ on $[n]$, makes at most two queries to both $D,D'$, and if $D=D'$, always accepts; otherwise, rejects with probability at least $3/4$. 
\end{theorem}
\begin{proof}
 Our algorithm generates two sets $S_1,S_2 \subseteq [n]$, both of them by including each element $i \in [n]$ independently with probability $1/2$. The algorithm  accepts if and only if  $D(S_1) = D'(S_1)$ and $D(S_2) = D'(S_2)$. Clearly if $D=D'$ then the algorithm always accepts. Next, we show that with probability at least $3/4$, either $D(S_1) \neq D'(S_1)$ or  $D(S_2) \neq  D'(S_2)$. 
 
 To prove the above, let us consider a random subset $S \subseteq [n]$ constructed by including each element $i \in [n]$ independently with probability $1/2$. Observe, it suffices to show that if $D \neq D'$, we have $D(S) = D'(S)$ with probability at most $1/2$. 

Let $U = \{i \in [n]\mid D(i) \neq D'(i)\}$ and $N = \{T \subseteq U \mid D(T) \neq D'(T)\}$. We claim that $|N| \ge \frac{2^{|U|}}{2}$. Note, if for any $A \subseteq U$, $D(A) = D'(A)$ then for any $j \in U \setminus A$, we have $D(A \cup \{j\}) \neq D'(A \cup \{j\})$. Consider an undirected graph with  $2^{U} :=\{V \mid V \subseteq U\}$ as the set of vertices and $\{\{V_1,V_2\} \mid V_1 \subseteq V_2, |V_2| = |V_1| +1, V_1,V_2 \subseteq U \}$ as the set of edges. Thus  $2^{U}\setminus N$ is an independent set (there is no edge between any  two vertices in $2^{U}\setminus N$)  in the above graph. Therefore, we have  $|N| \ge \frac{2^{|U|}}{2}$. Recall, $S \subseteq [n]$ is a random subset constructed by including each element $i \in [n]$ independently with probability $1/2$. Let $S_U = S \cap U$ and $S_E = S \cap ([n]\setminus U)$. Since  $|N| \ge \frac{2^{|U|}}{2},$ we have $D(S_U) \neq D'(S_U)$ with probability at least $1/2$. Also  note that $D(S_E) = D'(S_E)$ and hence we have $D(S) = D(S_E) + D(S_U) \neq D'(S_E) + D'(S_U) = D'(S)$ with probability at least $1/2$.
\end{proof}

\subsection*{Testing Grained Distributions}
A distribution is called $m$-grained if the probability of each element is an integer multiple of $\frac{1}{m}$. Goldreich and Ron~\cite{goldreich2021lower} show that testing if a  distribution on $[n]$ elements is $m$-grained, where $m = \Theta(n)$ requires $\Omega(n^c)$ {\Samp} queries for any constant $c < 1$. Grained distributions appear in several prior works, either implicitly (e.g.,~\cite{raskhodnikova2009strong}) or as in~\cite{goldreich2020uniform, goldreich2021lower}. 

Surprisingly, we show that just two {\SEval} queries are sufficient to test grained distributions. The algorithm and proof is similar to as that of Equivalence testing.

\begin{theorem}
There is an algorithm that given {\SEval} access to  distributions $D$ on $[n]$,  makes at most two queries, and if $D$ is $m$-grained, always accepts; otherwise, rejects with probability at least $3/4$. 
\end{theorem}
\begin{proof}
Our algorithm generates two sets $S_1,S_2 \subseteq [n]$, both of them by including each element $i \in [n]$ independently with probability $1/2$. The algorithm accepts if and only if both $D(S_1)$ and $D(S_2)$ are an integer multiple of $\frac{1}{m}$. Clearly if $D$ is $m$-grained then the algorithm always accepts. Below we claim that if $D$ is not $m$-grained, with probability at least $3/4$, either $D(S_1)$ or $D(S_2)$ is not an integer multiple of $\frac{1}{m}$.

Let $S \subseteq [n]$ be a random set  constructed by including each element $i \in [n]$ independently with probability $1/2$. Observe, it suffices to show that if $D$ is not $m$-grained then with probability at least $1/2$, $D(S)$ is not an integer multiple of $\frac{1}{m}$, the proof of which is similar to that in~\autoref{thm:seteval-equivalence}. Let $U = \{i \in [n]\mid D(i) \thinspace \text{is not an integer multiple of} \thinspace \frac{1}{m}\}$ and $N = \{T \subseteq U \mid P(T) \thinspace \text{is not an integer multiple of} \thinspace \frac{1}{m}\}$. We claim that $|N| \ge \frac{2^{|U|}}{2}$. If for any $A \subseteq U$, $D(A)$ is an integer multiple of $\frac{1}{m}$ then for any $j \in U \setminus A$,  $D(A \cup \{j\})$ is not an integer multiple of $\frac{1}{m}$. As argued in~\autoref{thm:seteval-equivalence}, we have  $|N| \ge \frac{2^{|U|}}{2}$.
Recall, $S \subseteq [n]$ is a random subset constructed by including each element $i \in [n]$ independently with probability $1/2$. Let $S_U = S \cap U$ and $S_E = S \cap ([n]\setminus U)$. Since $|N| \ge \frac{2^{|U|}}{2}$, $D(S_U)$ is not an integer multiple of $1/m$ with probability at least $1/2$. Also note that $D(S_E)$ is an integer multiple of $1/m$ and hence $D(S) = D(S_E) + D(S_U)$ is not an integer multiple of $1/m$ with probability at least $1/2$.
\end{proof}


\subsection*{Estimating $L_2$ norm}
The third example we consider is the classical (squared) $L_2$ norm estimation problem for a distribution. The $L_2$ norm of a distribution $D$ over $[n]$, denoted by $\ell_2(D)$, is defined as $\ell_2(D) := (\sum_{j \in [n]} D(j)^2)^{1/2}$. In this problem, given oracle access to a distribution $D$ on $[n]$, we are asked to estimate $\ell_2^2(D)$ up to a multiplicative $(1+\epsilon)$-factor. To the best of our knowledge, this problem has been studied previously only in the {\Samp} model~\cite{goldreich2011testing}, wherein it was shown that $\Omega(\frac{\sqrt{n}}{\epsilon^2})$ queries are required.

Surprisingly, we show that $\Omega(\sqrt{n})$ queries are also required in {\CO} and {\CE} models. Further, we show that in the {\SEval} model (hence {\CD} model too), only $O(1/\epsilon^2)$ queries are sufficient. Our algorithm (and also the analysis) resembles that of estimating $L_2$ norm in the streaming model~\cite{alon1999space}. We also show an almost tight lower bound of $\Omega(\frac{1}{\epsilon^2 \log \frac{1}{\epsilon}})$ for the {\CD} model.
\begin{theorem}
There is an algorithm that given {\SEval} access to  distributions $D$ on $[n]$ and $\epsilon > 0$, estimates $\ell_2^2(D)$ within a multiplicative $(1+\epsilon)$-factor with  probability at least $3/4$, by making at most $O(\frac{1}{\epsilon^2})$ queries.
\end{theorem}
\begin{proof}
Consider a $4$-wise independent hash function $h:[n] \rightarrow \{-1,1\}$. Let $S^+ = \{ j \in [n] \mid h(j) = 1\}$ and $S^-=\{j \in [n]\mid h(j) = -1\}$. Then we compute $X = (D(S_+) - D(S_-))^2$. We repeat the above procedure independently for $t = 4/\epsilon^2$ times and output the mean $\Bar{X}$ of the values of the random variable $X$ in each iteration. (Let $X_i$ denote the value of $X$ at the $i$-th iteration. Then we compute $\Bar{X} = \frac{\sum_{i=1}^t X_i}{t}$.) Note, the total number of {\SEval} queries is $8/\epsilon^2$.

We will show that at any iteration, $\E[X] = \sum_{j \in [n]}  D^2(j)$ and $Var[X] \le (\E[X])^2$. Then clearly we have $\E[\Bar{X}] = \E[X]$ and $Var[\Bar{X}] = \frac{Var[X]}{t}$. Hence, by Chebyshev inequality, we have 
\[
\Pr[|\Bar{X} - \sum_{j \in [n]}D(j)^2| > \epsilon \sum_{j \in [n]}D(j)^2] \le \frac{Var[\Bar{X}]}{\epsilon^2 (\E[\Bar{X}])^2} \le \frac{1}{4}.
\]

So, it only remains to prove that $\E[X] = \sum_{j \in [n]}  D^2(j)$ and $Var[X] \le (\E[X])^2$.
Note that $X = (D(S_+) - D(S_-))^2 = (\sum_{j \in [n]} h(j) D(j))^2 = \sum_{j \in [n]}  D^2(j) + 2 \sum_{i < j}h(i)h(j)D(i)D(j)$ (since $h^2(j) = 1$ for all $j \in [n]$). Therefore, $\E[X] =  \sum_{j \in [n]}  D^2(j) + 2 \sum_{i < j}\E[h(i)h(j)D(i)D(j)]$.  Since $h$ is $4$-wise independent, $\E[h(j)] = 0$ for any $j \in [n]$ and also $\E[h(i)h(j)] = \E[h(i)]\cdot \E[h(j)]$ = 0 for any $i \neq j$. Hence, $\E[X] = \sum_{j \in [n]}  D^2(j)$.  The variance $Var[X]$ can similarly be bounded as follows. Note that since $h$ is $4$-wise independent, $\E[h(i)h(j)h(k)h(l)] = \E[h(i)] \E[h(j)] \E[h(k)] \E[h(l)] = 0$ if any of $i,j,k,l \in [n]$ appears once or thrice.
\begin{align*}
    & Var[X] = \E[X^2] - (\E[X])^2\\
    & = \E[(\sum_{j \in [n]} h(j) D(j))^4] - (\sum_{j \in [n]}  D^2(j))^2 \\
    & = \sum_{i,j,k,l \in [n]^4} D(i)D(j)D(k)D(l) \E[h(i)h(j)h(k)h(l)] - (\sum_{j \in [n]}  D^2(j))^2 \\
    & = \sum_{j \in [n]}D^4(j) + \frac{1}{2} \binom{4}{2} \sum_{i\neq j} D^2(i) D^2(j) - (\sum_{j \in [n]}  D^2(j))^2\\
    & = \sum_{i\neq j} D^2(i) D^2(j) \le (\E[X])^2.
\end{align*}

\end{proof}

Next, we show that the above upper bound result for the squared $\ell_2$ norm estimation is almost tight.
\begin{theorem}
    \label{thm:lb-ceval-l2norm}
    Consider any $\alpha \in (0,1/2) $. Any algorithm that, given {\CD} access to a distribution $D$ on $[n]$ and any $\epsilon=n^{-\alpha}$, estimates $\ell_2^2(D)$ within a multiplicative $(1+\epsilon)$-factor with  probability at least $2/3$, must make $\Omega(\frac{1}{\epsilon^2 \log (1/\epsilon)})$ queries to the {\CD} oracle.
\end{theorem}
\begin{proof}
The proof is based on the observation that, in the proof of the Lemma~\ref{lem:lbepsilon}, we have  $\ell_2^2(D) = \sum_{j}D(j)^2 = \frac{2(|I_x|+|I_y|)-d_H(x,y)}{(|I_x|+|I_y|)^2}$. The ratio of $\ell_2^2(D)$ for Yes and No instances is 
    \[
\frac{2(|I_x|+|I_y|)-(n/2 - g)}{2(|I_x|+|I_y|)-(n/2+g)} \ge 1+\epsilon
    \]
    for the same values of parameters $\epsilon$ and $g$. Thus the lower bound also applies to estimates $\ell_2^2(D)$ within a multiplicative $(1+\epsilon)$-factor.
\end{proof}
\begin{theorem}
    \label{thm:lb-l2norm-conpr}
    Any algorithm that, given {\CE} ({\CO}) access to a distribution $D$ on $[n]$  estimates $\ell_2^2(D)$ within a multiplicative $5/4$-factor with  probability at least $2/3$, must make $\Tilde{\Omega}(\sqrt{n})$ queries to the {\CE} ({\CO}) oracle.
\end{theorem}
\begin{proof}
 Let $D_1$ be a uniform distribution on $[n]$, i.e, $D_1(j) = 1/n$ for all $j \in [n]$. Clearly,  $\ell_2^2(D_1) = 1/n$. 
For any $k \in [n]$ and any $G \subseteq [n]\setminus \{k\}$ such that $|G| = n-\sqrt{n}$, consider a distribution $D_{k,G}$ on $[n]$, defined as follows.
\begin{align*}
   D_{k,G}(j) =
\begin{cases}
\frac{1}{\sqrt{n}} &\text{ if } j = k\\
\frac{1}{n} &\text{ if } j \in G \\
0&\text{ otherwise}.
\end{cases} 
\end{align*}
Let $\mathcal{D}_2 = \bigcup_{ k,G}  D_{k,G}$ be the set of distributions over all possible values of $k$ and $G$. Note that for any $D_2 \in \mathcal{D}_2$, we have $\ell_2^2(D_2) = 1/n + (n-\sqrt{n})\cdot \frac{1}{n^2} \ge \frac{5}{4n}$.

Consider a distribution $D$ randomly selected from distributions in $D_1 \cup  \bigcup_{ k,G} D_{k,G}$ as follows. We pick $D_1$ with probability $1/2$ and any distribution $D_{k,G} \in \mathcal{D}_2$ with probability $\frac{1}{2 |\mathcal{D}_2|}$. As before, from Yao's theorem, it suffices to show $\Tilde{\Omega}(\sqrt{n})$ lower bound on the number of queries for any deterministic tester $T$ that correctly answers whether $D = D_1$ or $D \in \mathcal{D}_2$ with probability $2/3$, for randomly chosen distribution $D$ as described above.


For any fixed $S \subseteq [n]$, let $N^g_S$ be the random variable denoting the cardinality of the set $|\{j \in S: D(j) = 1/n \}|$ and $N^b_S$ be the random variable denoting the cardinality of the set $|\{j \in S: D(j) \in \{0,1/\sqrt{n}\} \}|$. 
 We want to show that $\frac{N^b_S}{N^b_S+N^g_S} > \frac{(\log n)^2}{\sqrt{n}}$ holds with probability at most $1/n^2$.
This is trivially true when $D = D_1$. Given that $D \in \mathcal{D}_2$, we have $\mu^g_S = \E[N^g_S] = |S|(1-1/\sqrt{n})$ and $\mu^b_S = \E[N^b_S] = |S|/\sqrt{n}.$ Further,  
\begin{enumerate}
    \item Suppose $|S| \ge \frac{\sqrt{n}}{(\log n)^2}$.  By Chernoff's bound,  we have $\Pr[N^b_S \ge (\log n)^2 \mu^b_S] \le exp(-\frac{1}{3}(\log n)^4 \mu^b_S)\le exp(-\frac{1}{3}(\log n)^2) < 1/n^2$ and $\Pr[N^g_S \le \mu^g_S/2] \le exp(-\frac{ \sqrt{n}}{10 (\log n)^2}) < 1/n^2$. Hence we have, $\frac{N^b_S}{N^b_S+N^g_S} > \frac{(\log n)^2}{\sqrt{n}}$ with probability at most $1/n^2$.
    \item If $|S| \le \frac{\sqrt{n}}{(\log n)^2}$ then again by Chernoff bound, we have  $\Pr[N^b_S \ge 1] \le exp(-(1/\mu^b_S -1)^2 \mu^b_S) \le exp(-(\log n)^2/4) < 1/n^2$. Hence , we have $\frac{N^b_S}{N^b_S+N^g_S} > \frac{(\log n)^2}{\sqrt{n}}$ with probability at most $1/n^2$.
\end{enumerate}
If $T$ makes $q < \frac{\sqrt{n}}{(\log n)^2}$ queries then probability that for all $q$ query sets $S$, we have  $\frac{N^b_S}{N^b_S+N^g_S} \le \frac{(\log n)^2}{\sqrt{n}}$  is at least $1 - q/n^2 \ge 1 - 1/n$. Further, probability that any of the $q$ {\CE} (or {\CO}) queries returns an element with probability value $1/\sqrt{n}$ is at most $q \cdot  \frac{(\log n)^2}{\sqrt{n}} < 1/3$ for $q < \frac{\sqrt{n}}{3 (\log n)^2}$. Hence, with probability at least $2/3$, the tester $T$ can not distinguish between $D = D_1$ or $D \in \mathcal{D}_2$ if  $q < \frac{\sqrt{n}}{3 (\log n)^2}$ (as the outcomes at every step will be an element with probability $1/n$ for both distributions).

\end{proof}
\section{Bounded-Set Conditioning}
\label{sec:bounded-cond}
In this section, we discuss the power of conditioning when the conditioned set is of bounded size, in light of the support size estimation problem. A \emph{$k$-bounded conditional}, in short ${\CO}_k$, oracle (similarly ${\CE}_k$ and ${\CD}_k$) allow the input set $S \subseteq [n]$ to be of size at most $k$. Additionally, we also allow algorithms to access {\Samp} oracle (to make sure that it is at least as powerful as the {\Samp} model). Note, this is indeed the case in the literature while allowing \emph{pair-conditioning} ({\PCo}), i.e., when $k=2$ (e.g.~\cite{chakraborty2016power, canonne2014testing, canonne2014aggregate}).

In this bounded conditioning model, we show that to get an estimation of the support size up to a constant factor, we need at least $\Omega(n/k)$ queries, which is in contrast with the $O(\log \log n)$ upper bounded in the standard (unbounded) conditioning model.
\begin{theorem}
\label{thm:boundedsetUB}
Any  algorithm that, given {\Samp} access and ${\CD}_k$ access to a distribution $D$ on $[n]$, approximates the support size $|\supp(D)|$ within a multiplicative $6/5$-factor with probability at least $2/3$, must make at least $\Omega(\frac{n}{k})$ queries.
\end{theorem}
\begin{proof}
Let us consider the following two sets of distributions. Let $\mathcal{D}_1 = \bigcup_{2 \le i \le n} \{D_i\}$, where $D_i$ is defined as follows.
\begin{align*}
   D_i(j) =
\begin{cases}
1 - \frac{2}{n} &\text{ if } j = 1\\
\frac{2}{n} &\text{ if } j=i\\
0&\text{ otherwise}.
\end{cases} 
\end{align*}
Let $\mathcal{D}_2 = \bigcup_{2 \le i < i' \le n} \{ D_{i,i'}\}$, where $D_{i,i'}$ is defined as follows.
\begin{align*}
   D_{i,i'}(j) =
\begin{cases}
1 - \frac{2}{n} &\text{ if } j = 1\\
\frac{1}{n} &\text{ if } j \in \{i,i'\} \\
0&\text{ otherwise}.
\end{cases} 
\end{align*}
Note that the support size of any distribution in $\mathcal{D}_1$ is $2$, whereas that of any distribution in $\mathcal{D}_2$ is $3$. Consider a distribution $\gamma$ that picks $r \in \{1,2\}$ uniformly at random. Then it samples a distribution from $\mathcal{D}_r$ uniformly at random.
Let $T$ be any deterministic algorithm that estimates the support size within a multiplicative $6/5$-factor with probability at least $2/3$ when a distribution $D$ is randomly chosen as per $\gamma$. By Yao's principle, it suffices to show that $T$ will make $\Omega(\frac{n}{k})$ queries. Observe, since $T$ estimates the support size within a multiplicative $6/5$-factor, one can easily determine the value of $r\in \{1,2\}$ from the output of $T$ (with probability at least $2/3$).

Without loss of generality, assume that the algorithm $T$, in the beginning, knows that $D(1) = 1-\frac{2}{n}$ for any $D \in \mathcal{D}_1 \cup \mathcal{D}_2$ (this is because the algorithm can first place a {\Samp} query to get the element 1 and then a ${\CD}_k$ query conditioning on $\{1\}$ to get to know the value of $D(1)$). Note that the outcome of {\Samp} is the first element with probability $1 - \frac{2}{n}$ for any distribution $D \in \mathcal{D}_1 \cup \mathcal{D}_2$. Thus if $T$ makes $t'$ {\Samp} queries then the probability that all of of them is 1 is $(1- \frac{2}{n})^{t'} \ge 1 - \frac{2t'}{n}$. Thus if $t' \le n/40$ then the algorithm $T$ always receives the first element as an outcome of all the {\Samp} queries with probability at least $\frac{19}{20}$. (Note, if $t' \ge n/40$ then the lower bound trivially holds.)

Let the sets queried by $T$ (on the ${\CD}_k$ oracle), in order, be $A_1,\dots,A_t$. Again, without loss of generality, assume that $A_j \subseteq \{2,\dots,n\}$ for all $j \in [t]$. Conditioned on the event that outcome of all the {\Samp} queries are the first element, the algorithm cannot determine $r \in \{1,2\}$ with probability (strictly) greater than $1/2$ if $D(A_j) = 0$ for all $j \in [t]$. We will later show that if $kt \le   \frac{n-2}{100}$ then the probability that $D(A_j) = 0$ for all $j \in [t]$ is at least $\frac{98}{100}$ (irrespective of $r \in \{1,2\}$). Therefore, using $T$, one can determine the correct value of $r$ with probability at most $\frac{1}{20}+\frac{98}{100} \cdot \frac{1}{2} + \frac{2}{100} \cdot 1 = \frac{56}{100} < \frac{2}{3}$, which leads to a contradiction. Hence, $t \ge \Omega(\frac{n}{k})$.

When $r = 1$, the probability that  $D(A_j) = 0$ for all $j \in [t]$ is at least $1 - \frac{kt}{n-1}$. Similarly, when $r = 2$, the probability that  $D(A_j) = 0$ for all $j \in [t]$ is at least $\frac{\binom{n-kt-1}{2}}{\binom{n-1}{2}} \ge (1 - \frac{kt}{n-2})^2$. So, irrespective of $r \in \{1,2\}$, if $kt \le \frac{n-2}{100}$, then the probability that $D(A_j) = 0$ for all $j \in [t]$ is at least $\frac{98}{100}$.
\end{proof}

Next, we show an upper bound of $O(\frac{n \log \log n}{k})$ queries in the ${\CO}_k$ model.

\begin{theorem}
\label{thm:boundedsetLB}
There is an  algorithm that, given ${\CO}_k$ access to a distribution $D$ on $[n]$, estimates the support size $|\supp(D)|$ within a multiplicative $6/5$-factor with  probability at least $2/3$, while making at most $O(\frac{n \log \log n}{k})$ queries.
\end{theorem}
\begin{proof}
Falahatgar et al.~\cite{falahatgar2016estimating} show that $O(\log \log n)$ queries are sufficient (with no restriction on size of the set for conditioning) to get a constant approximation, for an oracle access which given any (arbitrary sized) set $S \subseteq [n]$, returns whether $S \cap \supp(D) = \emptyset$ or not. 

Now we argue that oracle access with an arbitrary sized set $S$  can be simulated by $|S|/k$ many oracle access when conditioning on at most $k$-sized sets is allowed (i.e., using $|S|/k$ many ${\CO}_k$ queries). This can be done by partitioning the set $S$ into $|S|/k$ subsets where each subset is of size at most $k$ (in an arbitrary way). Observe, $S \cap \supp(D) = \emptyset$ if and only if all of these subsets do not intersect with $\supp(D)$, which can be determined by {\CO} query to each of these $k$-sized subsets. This immediately gives an upper bound of  $O(\frac{n \log \log n}{k})$.
\end{proof}
It is worth noting that in the above theorem, we argue that an oracle that on input $S$ just returns whether $S \cap \supp(D) = \emptyset$ or not, can be simulated by $|S|/k$ oracle access when conditioning on at most $k$-sized sets is allowed. (Note, using a {\CO} query, as defined in Canonne \emph{et al.}~\cite{canonne2014testing}, it is possible to determine whether $S \cap \supp(D) = \emptyset$ or not.) However, it is not clear whether it is possible to simulate any arbitrary {\CO} query using $|S|/k$ many ${\CO}_k$ queries. We now show that for stronger {\CD} model, a  {\CD} access to a set $S$  can \emph{always} be simulated by $|S|/k$ many oracle access when conditioning on at most $k$-sized sets is allowed.
\begin{lemma}
Given {\CD} access to a distribution $D$, a {\CD} oracle access on a set $S$ can be simulated by $|S|/k$ many ${\CD}_k$ oracle access.
\end{lemma}
\begin{proof}
Given a set $S$, we arbitrarily partition it into $|S|/k$ subsets each of size at most $k$. Let theses $|S|/k$ subsets be $S_1,\dots,S_{|S|/k}$. Recall, a {\CD} oracle for a distribution $D$ takes a set $S$ as input and returns a tuple $(j,D(j),\frac{D(j)}{D(S)})$ (where $j \in S$)  with probability $\frac{D(j)}{D(S)}$.  Note that returning a tuple $(j,D(j),\frac{D(j)}{D(S)})$ is equivalent to returning a tuple $(j,D(j),D(S))$. Now let $(j_i,D(j_i),D(S_i))$ be the outcome of ${\CD}_k$ on input the subset $S_i$ for $i \in [k]$. Note that  $D(S) = \sum_{i \in [|S|/k]}D(S_i)$.  If $D(S_i) = 0$ for all $i \in [k]$ then  obviously $D(S)= 0$ (and we are done as we know that $D(S) = 0$). So assume at least one of $D(S_i)$ is non-zero.  We now pick the $i$-th subset with probability $\frac{D(S_i)}{D(S)}$ and then return the tuple $(j_i,D(j_i),D(S))$ as the outcome for oracle access on input set $S$. It is easy to see that each $j_i$ is picked with probability $\frac{D(S_i)}{D(S)} \cdot \frac{D(j_i)}{D(S_i)} = \frac{D(j_i)}{D(S)}$, which concludes the proof.
\end{proof}

As a direct corollary, we get the following.
\begin{corollary}
Any algorithm with {\CD} query access to a distribution $D$ on $[n]$, that makes at most $t(n)$ queries, can be simulated using another algorithm with only ${\CD}_k$ query access to $D$, that makes at most $O(\frac{n}{k} t(n))$ queries.
\end{corollary}


\paragraph{Acknowledgments} The authors would like to thank anonymous reviewers for their useful suggestions and comments on an earlier version of this paper. Diptarka Chakraborty was supported in part by an MoE AcRF Tier 2 grant (WBS No. A-8000416-00-00).

\bibliography{references}

\newcommand{\etalchar}[1]{$^{#1}$}
\begin{thebibliography}{CFGM16}

\bibitem[ACK15]{acharya2015chasm}
Jayadev Acharya, Cl{\'e}ment~L Canonne, and Gautam Kamath.
\newblock A chasm between identity and equivalence testing with conditional
  queries.
\newblock In {\em Approximation, Randomization, and Combinatorial Optimization.
  Algorithms and Techniques (APPROX/RANDOM 2015)}. Schloss
  Dagstuhl-Leibniz-Zentrum fuer Informatik, 2015.

\bibitem[AMS99]{alon1999space}
Noga Alon, Yossi Matias, and Mario Szegedy.
\newblock The space complexity of approximating the frequency moments.
\newblock {\em Journal of Computer and system sciences}, 58(1):137--147, 1999.

\bibitem[BBM12]{blais2012property}
Eric Blais, Joshua Brody, and Kevin Matulef.
\newblock Property testing lower bounds via communication complexity.
\newblock {\em computational complexity}, 21(2):311--358, 2012.

\bibitem[BC18]{bhattacharyya2018property}
Rishiraj Bhattacharyya and Sourav Chakraborty.
\newblock Property testing of joint distributions using conditional samples.
\newblock {\em ACM Transactions on Computation Theory (TOCT)}, 10(4):1--20,
  2018.

\bibitem[BCG19]{blais2019distribution}
Eric Blais, Cl{\'e}ment~L Canonne, and Tom Gur.
\newblock Distribution testing lower bounds via reductions from communication
  complexity.
\newblock {\em ACM Transactions on Computation Theory (TOCT)}, 11(2):1--37,
  2019.

\bibitem[BCR{\etalchar{+}}10]{brody2010better}
Joshua Brody, Amit Chakrabarti, Oded Regev, Thomas Vidick, and Ronald~de Wolf.
\newblock Better gap-hamming lower bounds via better round elimination.
\newblock In {\em Approximation, Randomization, and Combinatorial Optimization.
  Algorithms and Techniques}, pages 476--489. Springer, 2010.

\bibitem[BDKR05]{batu2005complexity}
Tugkan Batu, Sanjoy Dasgupta, Ravi Kumar, and Ronitt Rubinfeld.
\newblock The complexity of approximating the entropy.
\newblock {\em SIAM Journal on Computing}, 35(1):132--150, 2005.

\bibitem[Ber76]{bernardo1976psi}
Jose~M Bernardo.
\newblock Algorithm as 103: Psi (digamma) function.
\newblock {\em Journal of the Royal Statistical Society. Series C (Applied
  Statistics)}, 25(3):315--317, 1976.

\bibitem[BM73]{blackwell1973ferguson}
David Blackwell and James~B MacQueen.
\newblock Ferguson distributions via p{\'o}lya urn schemes.
\newblock {\em The annals of statistics}, 1(2):353--355, 1973.

\bibitem[Can20]{canonne2020survey}
Cl{\'e}ment~L Canonne.
\newblock A survey on distribution testing: Your data is big. but is it blue?
\newblock {\em Theory of Computing}, pages 1--100, 2020.

\bibitem[CCK{\etalchar{+}}21]{canonne2021random}
Cl{\'e}ment~L Canonne, Xi~Chen, Gautam Kamath, Amit Levi, and Erik Waingarten.
\newblock Random restrictions of high dimensional distributions and uniformity
  testing with subcube conditioning.
\newblock In {\em Proceedings of the 2021 ACM-SIAM Symposium on Discrete
  Algorithms (SODA)}, pages 321--336. SIAM, 2021.

\bibitem[CFGM16]{chakraborty2016power}
Sourav Chakraborty, Eldar Fischer, Yonatan Goldhirsh, and Arie Matsliah.
\newblock On the power of conditional samples in distribution testing.
\newblock {\em SIAM Journal on Computing}, 45(4):1261--1296, 2016.

\bibitem[CJLW21]{chen2021learning}
Xi~Chen, Rajesh Jayaram, Amit Levi, and Erik Waingarten.
\newblock Learning and testing junta distributions with sub cube conditioning.
\newblock In {\em Conference on Learning Theory}, pages 1060--1113. PMLR, 2021.

\bibitem[CKOS15]{caferov2015optimal}
Cafer Caferov, Bar{\i}{\c{s}} Kaya, Ryan O’Donnell, and AC~Say.
\newblock Optimal bounds for estimating entropy with pmf queries.
\newblock In {\em International Symposium on Mathematical Foundations of
  Computer Science}, pages 187--198. Springer, 2015.

\bibitem[CM19]{chakraborty2019testing}
Sourav Chakraborty and Kuldeep~S Meel.
\newblock On testing of uniform samplers.
\newblock In {\em Proceedings of the AAAI Conference on Artificial
  Intelligence}, volume~33, pages 7777--7784, 2019.

\bibitem[CR12]{chakrabarti2012optimal}
Amit Chakrabarti and Oded Regev.
\newblock An optimal lower bound on the communication complexity of
  gap-hamming-distance.
\newblock {\em SIAM Journal on Computing}, 41(5):1299--1317, 2012.

\bibitem[CR14]{canonne2014aggregate}
Cl{\'e}ment Canonne and Ronitt Rubinfeld.
\newblock Testing probability distributions underlying aggregated data.
\newblock In {\em International Colloquium on Automata, Languages, and
  Programming}, pages 283--295. Springer, 2014.

\bibitem[CRS15]{canonne2014testing}
Cl{\'e}ment~L Canonne, Dana Ron, and Rocco~A Servedio.
\newblock Testing probability distributions using conditional samples.
\newblock {\em SIAM Journal on Computing}, 44(3):540--616, 2015.

\bibitem[CT06]{CT06}
Thomas~M. Cover and Joy~A. Thomas.
\newblock {\em Elements of information theory (2. ed.)}.
\newblock Wiley, 2006.

\bibitem[Dav59]{davis1959leonhard}
Philip~J Davis.
\newblock Leonhard euler's integral: A historical profile of the gamma
  function: In memoriam: Milton abramowitz.
\newblock {\em The American Mathematical Monthly}, 66(10):849--869, 1959.

\bibitem[DM22]{DM22}
Remi Delannoy and Kuldeep~S Meel.
\newblock On almost-uniform generation of sat solutions: The power of 3-wise
  independent hashing.
\newblock 2022.

\bibitem[FJO{\etalchar{+}}15]{falahatgar2015faster}
Moein Falahatgar, Ashkan Jafarpour, Alon Orlitsky, Venkatadheeraj Pichapati,
  and Ananda~Theertha Suresh.
\newblock Faster algorithms for testing under conditional sampling.
\newblock In {\em Conference on Learning Theory}, pages 607--636. PMLR, 2015.

\bibitem[FJO{\etalchar{+}}16]{falahatgar2016estimating}
Moein Falahatgar, Ashkan Jafarpour, Alon Orlitsky, Venkatadheeraj Pichapati,
  and Ananda~Theertha Suresh.
\newblock Estimating the number of defectives with group testing.
\newblock In {\em 2016 IEEE International Symposium on Information Theory
  (ISIT)}, pages 1376--1380. IEEE, 2016.

\bibitem[GHNR14]{GHNR14}
Andrew~D. Gordon, Thomas~A. Henzinger, Aditya~V. Nori, and Sriram~K. Rajamani.
\newblock Probabilistic programming.
\newblock In {\em Future of Software Engineering Proceedings}, FOSE 2014, page
  167–181, New York, NY, USA, 2014. Association for Computing Machinery.
\newblock \href {https://doi.org/10.1145/2593882.2593900}
  {\path{doi:10.1145/2593882.2593900}}.

\bibitem[GJM22]{GJM22}
Priyanka Golia, Brendan Juba, and Kuldeep~S. Meel.
\newblock Efficient entropy estimation with applications to quantitative
  information flow.
\newblock In {\em International Conference on Computer-Aided Verification
  (CAV)}, 2022.

\bibitem[GMV09]{guha2009sublinear}
Sudipto Guha, Andrew McGregor, and Suresh Venkatasubramanian.
\newblock Sublinear estimation of entropy and information distances.
\newblock {\em ACM Transactions on Algorithms (TALG)}, 5(4):1--16, 2009.

\bibitem[Gol20]{goldreich2020uniform}
Oded Goldreich.
\newblock The uniform distribution is complete with respect to testing identity
  to a fixed distribution.
\newblock In {\em Computational Complexity and Property Testing}, pages
  152--172. Springer, 2020.

\bibitem[GR11]{goldreich2011testing}
Oded Goldreich and Dana Ron.
\newblock On testing expansion in bounded-degree graphs.
\newblock In {\em Studies in Complexity and Cryptography. Miscellanea on the
  Interplay between Randomness and Computation}, pages 68--75. Springer, 2011.

\bibitem[GR21]{goldreich2021lower}
Oded Goldreich and Dana Ron.
\newblock Lower bounds on the complexity of testing grained distributions.
\newblock Technical Report TR21-129, Electronic Colloquium on Computational
  Complexity~…, 2021.

\bibitem[JKB95]{johnson1995continuous}
Norman~L Johnson, Samuel Kotz, and Narayanaswamy Balakrishnan.
\newblock {\em Continuous univariate distributions, volume 2}, volume 289.
\newblock John wiley \& sons, 1995.

\bibitem[KT19]{kamath2019anaconda}
Gautam Kamath and Christos Tzamos.
\newblock Anaconda: A non-adaptive conditional sampling algorithm for
  distribution testing.
\newblock In {\em Proceedings of the Thirtieth Annual ACM-SIAM Symposium on
  Discrete Algorithms}, pages 679--693. SIAM, 2019.

\bibitem[MPC20]{meel2020testing}
Kuldeep~S Meel, Yash~Pralhad Pote, and Sourav Chakraborty.
\newblock On testing of samplers.
\newblock {\em Advances in Neural Information Processing Systems},
  33:5753--5763, 2020.

\bibitem[Nar21]{Narayanan21}
Shyam Narayanan.
\newblock On tolerant distribution testing in the conditional sampling model.
\newblock In D{\'{a}}niel Marx, editor, {\em Proceedings of the 2021 {ACM-SIAM}
  Symposium on Discrete Algorithms, {SODA} 2021, Virtual Conference, January 10
  - 13, 2021}, pages 357--373. {SIAM}, 2021.

\bibitem[OS18]{onak2018probability}
Krzysztof Onak and Xiaorui Sun.
\newblock Probability--revealing samples.
\newblock In {\em International Conference on Artificial Intelligence and
  Statistics}. PMLR, 2018.

\bibitem[RRSS09]{raskhodnikova2009strong}
Sofya Raskhodnikova, Dana Ron, Amir Shpilka, and Adam Smith.
\newblock Strong lower bounds for approximating distribution support size and
  the distinct elements problem.
\newblock {\em SIAM Journal on Computing}, 39(3):813--842, 2009.

\bibitem[RS09]{rubinfeld2009testing}
Ronitt Rubinfeld and Rocco~A Servedio.
\newblock Testing monotone high-dimensional distributions.
\newblock {\em Random Structures \& Algorithms}, 34(1):24--44, 2009.

\bibitem[SC19]{scarlett2019introductory}
Jonathan Scarlett and Volkan Cevher.
\newblock An introductory guide to {Fano}'s inequality with applications in
  statistical estimation.
\newblock {\em arXiv preprint arXiv:1901.00555}, 2019.

\bibitem[Sha48]{Sha48}
C.~E. Shannon.
\newblock {A mathematical theory of communication}.
\newblock {\em Bell system technical journal}, 27, 1948.

\bibitem[VV11]{valiant2011estimating}
Gregory Valiant and Paul Valiant.
\newblock Estimating the unseen: an n/log (n)-sample estimator for entropy and
  support size, shown optimal via new clts.
\newblock In {\em Proceedings of the forty-third annual ACM symposium on Theory
  of computing}, pages 685--694, 2011.

\bibitem[WY19]{wu2019chebyshev}
Yihong Wu and Pengkun Yang.
\newblock Chebyshev polynomials, moment matching, and optimal estimation of the
  unseen.
\newblock {\em The Annals of Statistics}, 47(2):857--883, 2019.

\bibitem[Yao77]{yao1977probabilistic}
Andrew Chi-Chin Yao.
\newblock Probabilistic computations: Toward a unified measure of complexity.
\newblock In {\em 18th Annual Symposium on Foundations of Computer Science
  (sfcs 1977)}, pages 222--227. IEEE Computer Society, 1977.

\end{thebibliography}


\end{document}